%% file: The_SPARSE-Relativization_Framework_and_Applications_to_Optimal_Proof_Systems.tex
\pdfoutput=1
\documentclass[a4paper,UKenglish,thm-restate,pdfa]{lipics-v2021-modified}

\usepackage{mathtools}					
\usepackage{xargs}                      
\usepackage[pdftex,dvipsnames]{xcolor}  
\usepackage{algorithm}
\usepackage{algpseudocode}
\usepackage[colorinlistoftodos,prependcaption,textsize=tiny]{todonotes}

\input{9-Standardnotation.tex}
\newcommandx{\unsure}[2][1=]{\todo[linecolor=red,backgroundcolor=red!25,bordercolor=red,#1]{#2}}
\newcommandx{\change}[2][1=]{\todo[linecolor=blue,backgroundcolor=blue!25,bordercolor=blue,#1]{#2}}
\newcommandx{\info}[2][1=]{\todo[linecolor=OliveGreen,backgroundcolor=OliveGreen!25,bordercolor=OliveGreen,#1]{#2}}
\newcommandx{\improvement}[2][1=]{\todo[linecolor=Plum,backgroundcolor=Plum!25,bordercolor=Plum,#1]{#2}}
\newcommandx{\thiswillnotshow}[2][1=]{\todo[disable,#1]{#2}}
\algdef{SE}[SUBALG]{Indent}{EndIndent}{}{\algorithmicend\ }%
\algtext*{Indent}
\algtext*{EndIndent}
\hideLIPIcs
\title{The SPARSE-Relativization Framework and Applications to Optimal Proof Systems}
\titlerunning{SPARSE-Relativization Framework and Optimal Proof Systems}

\author{Fabian Egidy}{University of Würzburg, Germany}{fabian.egidy@uni-wuerzburg.de}{https://orcid.org/0000-0001-8370-9717}{supported by the German Academic Scholarship Foundation.}
\authorrunning{Fabian Egidy}
\Copyright{Fabian Egidy}
\nolinenumbers
\acknowledgements{We thank Christian Glaßer and Katharina Pfab for their helpful feedback and suggestions.}

\begin{document}
\maketitle
\begin{abstract}
    We investigate the following longstanding open questions raised by Krajíček and Pudlák (J.~Symb.~L.~1989), Sadowski (FCT 1997), Köbler and Messner (CCC 1998) and Messner (PhD 2000).
    \medskip
    \begin{itemize}
        \item[] Q1: Does $\TAUT$ have (p-)optimal proof systems?
        \item[] Q2: Does $\QBF$ have (p-)optimal proof systems?
        \item[] Q3: Are there arbitrarily complex sets with (p-)optimal proof systems?
    \end{itemize}

    Recently, Egidy and Glaßer (STOC 2025) contributed to these questions by constructing oracles that show that there are no relativizable proofs for positive answers of these questions, even when assuming well-established conjectures about the separation of complexity classes. We continue this line of research by constructing an oracle that shows that there are no relativizable proofs for negative answers of these questions, even when assuming well-established conjectures about the separation of complexity classes. For this, we introduce the $\SPARSE$-relativization framework, which is an application of the notion of bounded relativization by Hirahara, Lu, and Ren (CCC 2023). This framework allows the construction of sparse oracles for statements such that additional useful properties (like an infinite polynomial-time hierarchy) hold. 
    
    By applying the $\SPARSE$-relativization framework, we show that the oracle construction of Egidy and Glaßer also yields the following new oracle.
    \medskip
    \begin{itemize}
        \item[] $O_1$: No set in $\PSPACE \setminus \NP$ has optimal proof systems, $\NP \subsetneq \PH \subsetneq \PSPACE$, and $\PH$ collapses
    \end{itemize}
    \medskip
    \noindent
    Cook and Krajíček (J.~Symb.~L.~2007) and Beyersdorff, Köbler, and Müller (Inf.~Comp.~2011) construct optimal proof systems for all recursively enumerable sets when the proof systems have access to one bit of non-uniform advice. We translate their techniques to the oracle setting and apply our $\SPARSE$-relativization framework to obtain the following new oracle. 
    \medskip
    \begin{itemize}
        \item[] $O_2$: All sets in $\PSPACE$ have p-optimal proof systems, there are arbitrarily complex sets with 
        \phantom{$O_2$: }p-optimal proof systems, and $\PH$ is infinite
    \end{itemize}
    \medskip
    Together with previous results, our oracles show that questions Q1 and Q2 are independent of an infinite or collapsing polynomial-time hierarchy. This contributes to Pudlák's (Bull.~Symb.~Log.~2017) search for plausible conjectures relativizably implying a negative answer to Q1. In particular, oracle $O_2$ rules out an infinite polynomial-time hierarchy as such a conjecture.
\end{abstract}

\section{Introduction}
We investigate the relationship between the existence of (p-)optimal proof systems for sets and separations of complexity classes. Cook and Reckhow \cite{cr79} define a proof system for a set $L$ as a polynomial-time computable function $f$ whose range is $L$. They motivate the study of proof systems by their connection to the separation of complexity classes, because they show: $\NP = \coNP$ if and only if there are proof systems with polynomial-size proofs for the set of propositional tautologies $\TAUT$. Hence, ruling out polynomial-size proofs for all proof systems for $\TAUT$ suffices to show $\NP \neq \coNP$. This approach to tackle the $\NP$ versus $\coNP$ question is known as the Cook-Reckhow program \cite{bus96}.

Krajíček and Pudlák \cite{kp89} provide a central definition for the Cook-Reckhow program by capturing the idea of a most powerful proof system. A proof system for a set $L$ is optimal if its proofs are at most polynomially longer than proofs in any other proof system for $L$. A proof system $f$ for a set $L$ is p-optimal if proofs from any other proof system for $L$ can be translated in polynomial-time into proofs in $f$. For the Cook-Reckhow program, the existence of an optimal proof system $f$ for $\TAUT$ means that it suffices to consider lower bounds only for $f$ instead of all proof systems for $\TAUT$.

Krajíček and Pudlák \cite{kp89} raise the question of whether $\TAUT$ has (p-)optimal proof systems. Subsequently, Sadowski \cite{sad97} raises the same question for the set of true quantified Boolean formulas $\QBF$ and Köbler and Messner \cite{km98} and Messner \cite{Mes00} explicitly raise the question for any set outside $\NP$, respectively $\P$\footnote{It is obvious that all nonempty sets in $\NP$ (resp., $\P$) have optimal (resp., p-optimal) proof systems.}. Here, we are interested in the following longstanding open questions.
\begin{spreadlines}{2ex}
\begin{align*}
&\text{\textbf{Q1}: {\em Does $\TAUT$ have (p-)optimal proof systems?}}
\\
&\text{\textbf{Q2}: {\em Does $\QBF$ have (p-)optimal proof systems?}}
\\
&\text{\textbf{Q3}: {\em Are there arbitrarily complex sets with (p-)optimal proof systems?}}
\end{align*}
\end{spreadlines}

Recently, Egidy and Glaßer \cite{eg25} investigated these questions and provided oracles that show that there are no relativizable proofs for positive answers of these questions, even when assuming well-established conjectures, namely an infinite $\PH$ for Q1 and Q2, and $\NP \neq \coNP$ for Q3. We continue this line of research by constructing an oracle that shows that there are no relativizable proofs for negative answers of these questions, even when assuming an infinite $\PH$. We do this by introducing the $\SPARSE$-relativization framework, which is an application of the notion of bounded relativization by Hirahara, Lu, and Ren \cite{hlr23}. Before we make our contribution precise, we give an overview on (p-)optimal proof systems in complexity theory.

\subsection{Previous Work on (P-)Optimal Proof Systems}
\subparagraph{Separation of complexity classes.} 
Krajíček and Pudlák \cite{kp89} connect (p-)optimal proof systems to the separation of complexity classes by showing that if optimal proof systems for $\TAUT$ do not exist, then $\NE \neq \coNE$, and that if p-optimal proof systems for $\TAUT$ do not exist, then $\E \neq \NE$. Köbler, Messner, and Torán \cite{kmt03} improve both results to $\NEE \neq \coNEE$ for optimal and $\EE \neq \NEE$ for p-optimal proof systems for $\TAUT$. 
Chen, Flum, and Müller \cite{cfm14} show that the nonexistence of optimal proof systems in $\NE \setminus \NP$ implies  the Measure Hypothesis \cite{lut97} to fail, i.e., $\NP$ has measure $0$ in $\E$. 

It is an easy observation that all sets in $\NP$ (resp., $\P$) have (p-)optimal proof systems. We do not know a single set outside $\NP$ (resp., $\P$) that has (p-)optimal proof systems. However, Messner \cite{mes99,Mes00} shows that all sets hard for $\E$ do not have p-optimal proof systems and all sets hard for $\coNE$ do not have optimal proof systems. 
Furthermore, Messner \cite{Mes00} shows that $\NP = \coNP$ (resp., $\P = \NP$) implies a positive answer to Q3 for (p-)optimal proof systems.

\subparagraph{Almost optimal algorithms.} Krajíček and Pudlák \cite{kp89} show that $\TAUT$ has p-optimal proof systems if and only if there exists an almost optimal algorithm (an algorithm that is optimal on positive instances, also called optimal acceptor) for $\TAUT$. Sadowski \cite{sad99} transfers this result to $\SAT$ (the set of satisfiable Boolean formulas) and Messner \cite{mes99,Mes00} generalizes both to all paddable sets. Furthermore, Messner shows that $\P = \NP$ implies that a set has an optimal acceptor if and only if it has a p-optimal proof system. 

\subparagraph{Promise classes.} Optimal proof systems are closely related to promise classes. Razborov \cite{raz94} shows that the existence of optimal proof systems for $\TAUT$ implies the existence of complete sets in the class $\DisjNP$ (the class of disjoint $\NP$-pairs). Similarly, Sadowski \cite{sad97} shows that the existence of p-optimal proof systems for $\QBF$ implies complete sets for $\NP \cap \coNP$. Köbler, Messner, and Torán \cite{kmt03} discover a general connection between proof systems and complete sets for promise classes. They improve and generalize past results by showing that the existence of optimal and p-optimal proof systems for various sets of the polynomial-time hierarchy implies the existence of complete sets for promise classes and probabilistic classes. For example, p-optimal proof systems for $\TAUT$ imply complete sets for $\UP$ and $\NP \cap \SPARSE$, and p-optimal proof systems for complete sets of the second level of the polynomial-time hierarchy imply complete sets for $\BPP$, $\ZPP$, and $\RP$ \cite{kmt03}. Beyersdorff, Köbler, and Messner \cite{bkm09} and Pudlák \cite{pud17} connect p-optimal proof systems and function classes. Together they show that p-optimal proof systems for $\SAT$ imply complete sets for $\TFNP$. 

Pudlák \cite{pud17} initiates a research program to either prove further such implications or to disprove them relative to an oracle. Meanwhile, for all pairs of conjectures except one that Pudlák is considering, we either have a proof for the implication or a refutation of the implication relative to an oracle. Hence, there are several oracles (see \cite{deg24, dg20, dos20a, dos20b, eeg22, gssz04, kha22}) that separate conjectures about the existence of (p-)optimal proof systems from other conjectures about complete sets in promise classes.

\subsection{Our Contribution}

\subsubsection{The SPARSE-Relativization Framework}
Egidy and Glaßer \cite{eg25} combine a sparse oracle with an oracle relative to which the polynomial-time hierarchy is infinite such that the polynomial-time hierarchy remains infinite relative to the combined oracle. The following theorem states the main reason why this approach works. 
\begin{theorem}[\cite{bbs86,ls86}]\label{thm:ph-inf-intro}
    If there exists a sparse set $S$ such that the polynomial-time hierarchy relative
to $S$ collapses, then the (unrelativized) polynomial-time hierarchy collapses.
\end{theorem}

So when the polynomial-time hierarchy is infinite, then it is infinite relative to all sparse oracles. We introduce the notion of \emph{sparse-oracle-independent} statements, 
referring to complexity-theoretic statements that are true if and only if they are also true relative to any sparse oracle. Note that the polynomial-time hierarchy being infinite is sparse-oracle-independent by Theorem \ref{thm:ph-inf-intro} and that there are more statements of this type (cf.~Theorem~\ref{thm:sparse-relativizing-results}). 

We generalize the approach of Egidy and Glaßer by developing a framework to translate a sparse oracle for some statement $S$ into an oracle relative to which $S$ and a sparse-oracle-independent statement\footnote{More precisely, this is possible for those sparse-oracle-independent statements that hold relative to at least one oracle.} holds. In this way, the new framework provides stronger proof barriers. 

The framework's main advantages lie in its simplicity and wide applicability. In general it is very difficult to combine two existing oracle constructions with each other. However, our framework makes it possible to take two independent oracle constructions and easily perform such combinations in certain situations. More precisely, whenever a statement holds relative to a sparse oracle, our framework can probably enhance the sparse oracle by complex additional properties at little cost.
Section \ref{sec:sparse-relativization} gives an in-depth explanation of the framework. There, we also provide a simple proof for two new results that would have been difficult to achieve without the application of the framework.

\subsubsection{Contribution to Q1, Q2, and Q3}
We apply the $\SPARSE$-relativization framework and show that the oracle construction of Egidy and Glaßer \cite{eg25} can be combined with an oracle by Ko \cite{ko89} yielding the following oracle.
\begin{spreadlines}{0ex}
\begin{align*}
    &O_1\colon \text{ No set in $\PSPACE \setminus \NP$ has optimal proof systems, $\NP \subsetneq \PH \subsetneq \PSPACE$, and}
    \\
    \vspace{-100em}&\phantom{O_1\colon} \text{ $\PH$ collapses}
\end{align*}
\end{spreadlines}
Furthermore, we construct a tally oracle relative to which Q1 and Q2 are answered in the positive. Since the oracle is constructed using the $\SPARSE$-relativization framework, we obtain an oracle $O_2$ relative to which the following holds.
\begin{spreadlines}{0ex}
\begin{align*}
    &O_2\colon \text{ All sets in $\PSPACE$ have p-optimal proof systems, there are arbitrarily complex} 
    \\
    \vspace{-100em}&\phantom{O_2\colon} \text{ sets with p-optimal proof systems, and $\PH$ is infinite}
\end{align*}
\end{spreadlines}

\subparagraph{Contribution to Q1 and Q2.} By the oracle of Baker, Gill, and Solovay \cite{bgs75} relative to which $\P = \PSPACE$ and by an oracle of Egidy and Glaßer \cite{eg25}, there are no relativizable proofs of the following implications.
\begin{alignat*}{2}
    \PH \text{ collapses} &\implies \text{negative answer to Q1 and Q2} \qquad && \text{\cite{bgs75}}\\
    \PH \text{ infinite} &\implies \text{positive answer to Q1 and Q2} \qquad && \text{\cite{eg25}}
\end{alignat*}
Our oracles show that the remaining two implications also have no relativizable proofs.
\begin{alignat*}{2}
    \PH \text{ collapses} &\implies \text{positive answer to Q1 and Q2} \qquad && \text{(by $O_1$ from Theorem \ref{thm:conn-and-ph})}\\
    \PH \text{ infinite} &\implies \text{negative answer to Q1 and Q2} \qquad && \text{(by $O_2$ from Corollary \ref{cor:O2})}
\end{alignat*}
Note that since it is believed that the polynomial-time hierarchy is infinite and that Q1 and Q2 both have negative answers \cite{pud17}, both of our results rule out relativizable proofs of a ``likely premise'' implying a ``likely consequence''.  In total, regarding relativizable proof techniques, the questions Q1 and Q2 are independent of an infinite or collapsing polynomial-time hierarchy. Moreover, Pudlák \cite{pud17} searches for plausible conjectures relativizably implying a negative answer to Q1. Oracle $O_2$ rules out an infinite polynomial-time hierarchy as such a conjecture.

\subparagraph{Contribution to Q3.}
Egidy and Glaßer \cite{eg25} construct an oracle that shows that there is no relativizable proof for a positive answer of Q3, even when assuming $\NP \neq \coNP$. Messner \cite{mes99,Mes00} shows that answering Q3 in the negative implies $\NP \neq \coNP$. Oracle $O_2$ extends the barrier of Messner by showing that there is no relativizable proof for a negative answer of Q3, even when assuming an infinite polynomial-time hierarchy. This shows that proving Q3 in the negative has some inherent difficulty that does not only stem from its relation to the $\NP$ versus $\coNP$ question.

\subsection{Comparing Sparse and Random Oracles}
Our $\SPARSE$-relativization framework is in a similar spirit as Fortnow's \cite{for99} ``Book Trick'',
thus we give a brief introduction. Like ourselves, Fortnow is interested in relativized worlds where some propositions hold along with an infinite polynomial-time hierarchy. He also observes that ``creating such worlds is still a difficult task'' and achieving an infinite polynomial-time hierarchy by hand using techniques of Yao \cite{yao85} or H{\aa}stad \cite{has89} are ``complicated at best''. Hence, he introduces the Book Trick, which is based on a result by Book \cite{boo94}.
\begin{corollary}[\cite{boo94}]
    For every oracle $A$, if the polynomial-time hierarchy is infinite relative to $A$ then for random $R$, the polynomial-time hierarchy relative to $A \oplus R$ is infinite.    
\end{corollary}
\begin{lemma}[\cite{for99}, The Book Trick]
    If $P$ is a relativizable proposition such that for all $A$ and random $R$, $P^{A \oplus R}$ is true then there exists an oracle $B$ such that $P^B$ holds and the polynomial-time hierarchy is infinite relative to $B$.
\end{lemma}
Using the Book Trick, Fortnow \cite{for99} comes up with several oracles relative to which the polynomial-time hierarchy is infinite and relative to which propositions that hold relative to random oracles are true (similar to how we use the $\SPARSE$-relativization framework in Theorem \ref{thm:conn-and-ph}). However, the breakthrough result by H{\aa}stad, Rossman, Servedio, and Tan \cite{hrst17} showing that the polynomial-time hierarchy is infinite relative to a random oracle made the Book Trick obsolete. It is clear that any pair of propositions that hold individually relative to a random oracle also hold jointly relative to a random oracle. 

In total, this gives us two comparably elegant approaches when we are interested in an oracle with some property $P$ and an infinite polynomial-time hierarchy: either show that $P$ holds relative to a random oracle or show that $P$ holds relative to a sparse oracle\footnote{Joined with an arbitrary base oracle as explained in Section \ref{sec:sparse-relativization}.}. Note that the random oracle approach does not work for a collapsing polynomial-time hierarchy (also often a desirable property), in contrast to our $\SPARSE$-relativization framework (cf.~Theorem \ref{thm:bbs86}).

\subsection{Sketch of the Construction of $\text{O}_\text{2}$}
The main difficulty is to construct a sparse oracle relative to which (a) all sets in $\PSPACE$ have p-optimal proof systems and (b) there are arbitrarily complex sets with p-optimal proof systems. Then the $\SPARSE$-relativization framework provides additionally an infinite polynomial-time hierarchy.

To (a): Cook and Krajíček \cite{ck07} show how to obtain optimal proof systems for $\TAUT$ when those are given one bit of non-uniform advice. Beyersdorff, Köbler, and Müller \cite{bkm11} generalize this to all recursively enumerable sets. For a set $L$ it works like this: Let $f$ be an arbitrary polynomial-time computable function and $n \in \N$. If for all $x \in \Sigma^n$ it holds that $f(x) \in L$, then choose the advice bit for length $\langle \textrm{code}(f), n \rangle$ as 1, otherwise as 0. A p-optimal proof system can rely on these advice bits to decide which computations it can simulate without outputting an element outside $L$. We translate this idea to oracles. Note that this translation is not straightforward, because the advice setting is asymmetric: proof systems have acces to one bit of non-uniform advice, while the sets being proved are recognized by machines without advice. In contrast, our oracle setting is symmetric: both proof systems and machines recognizing the proved sets have access to the same oracle.\footnote{This distinction becomes evident when considering $\coNE$-hard sets, which do not have optimal proof systems as shown by Messner \cite{mes99}. Messner's result relativizes, i.e., it also holds when $\coNE$-machines and proof systems have access to an arbitrary oracle. On the other hand, by Beyersdorff, Köbler, and Müller \cite{bkm11}, there is an optimal proof system with one bit of non-uniform advice for any $\coNE$-hard set.}
Since the $\PSPACE$-machines and proof systems can only pose polynomially bounded queries, we can hide advice in the oracle from the computations it encodes using polynomially long paddings. This still allows the p-optimal proof system to rely on advice and only results in polynomially longer proofs. 

To (b): For every time-constructible function $t$, we define a set $L_t$ that contains $0^n$ if some code word of length $\approx t(n)$ is inside the oracle. We diagonalize against any machine trying to decide $L_t$ without querying words of length $\approx t(n)$, i.e., we add the code word of length $\approx t(n)$ to the oracle if and only if such a machine rejects $0^n$. By this, any machine deciding $L_t$ must have runtime at least $t$. Since also runtime $t$ suffices to decide $L_t$ and the proofs for $x \in L_t$ are generic (a fixed code word is in the oracle), the existence of p-optimal proof systems for $L_t$ follows.

\section{Preliminaries}
Let $\Sigma \coloneqq \{0,1\}$ be the default alphabet and $\Sigma^*$ be the set of finite words over $\Sigma$. The set of all (positive) natural numbers is denoted by $\N$ ($\N^+$). For $a,b \in \N$, we define $a\N+b \coloneqq \{a \cdot n + b \mid n \in \N\}$. We write the empty set as $\emptyset$. The cardinality of a set $A$ is denoted by $\card{A}$. For a set $A \subseteq \Sigma^*$ and a number $n \in \N$, we define $A^{\leq n} \coloneqq \{w \in A \mid |w| \leq n\}$ and analogous for $=$. For a clearer notation we use $\Sigma ^{\leq n}$ for ${\Sigma^*}^{\leq n}$ and $\Sigma^n$ for ${\Sigma^*}^{=n}$. The operators $\cup$, $\cap$, and $\setminus$ denote the union, intersection and set-difference. We denote the marked union of sets $A,B$ by $A \oplus B \coloneqq 2A \cup (2B+1)$. The domain and range of a function $f$ are denoted by $\dom (f)$ and $\ran (f)$. We use the notion of time-constructible functions as described in \cite{ab09}. We use the default model of a Turing machine in both the deterministic and non-deterministic variant. The language decided by a Turing machine $M$ is denoted by $L(M)$. We use Turing transducer to compute functions. For a Turing transducer $F$ we write $F(x)=y$ when on input $x$ the transducer outputs $y$. Hence, a Turing transducer $F$ computes a function and we may denote ''the function computed by $F$'' by $F$ itself.

\subparagraph*{Complexity classes.}
The definition of the basic complexity classes such as $\P$, $\FP$, $\NP$, $\coNP$, $\Sigma_k^p$, $\Pi_k^p$, $\PH$, $\IP$, $\PSPACE$, $\E$, $\DTIME$ and complete sets $\TAUT$ and $\QBF$ can be found in \cite{ab09}. The class $\EXPH$ is rigorously defined in \cite{hlr23}. A set $S \subseteq \Sigma^*$ is sparse if there is a polynomial $p$ such that $\card{S^{\leq n}} \leq p(n)$ for all $n \in \N$. A set $T \subseteq 0^*$ is tally. We denote the class of all sparse sets as $\SPARSE$ and the class of all tally sets as $\TALLY$. When we use the expression ``there are arbitrarily complex sets that \dots'' then this refers to a collection $\mathcal{C}$ of sets such that for any time-constructible function $t$ there is some decidable set $A \in \mathcal{C}$ with $A \notin \DTIME(t)$. 

\subparagraph*{Proof systems.}
We use proof systems for sets defined by Cook and Reckhow \cite{cr79}. They define a function $f \in \FP$ to be a proof system for $\ran (f)$. Furthermore:
\begin{itemize}
\item A proof system $g$ is (p-)simulated by a proof system $f$, denoted by $g \leq f$ (resp., $g \psim f$), if there exists a total function $\pi$ (resp., $\pi \in \FP$) and a polynomial $p$ such that $|\pi(x)| \leq p(|x|)$ and $f(\pi(x)) = g(x)$ for all $x \in \Sigma^*$. In this context the function $\pi$ is called simulation function. Note that $g \psim f$ implies $g \leq f$.
\item A proof system $f$ is (p-)optimal for $\ran(f)$, if $g \leq f$ (resp., $g \psim f$) for all $g \in \FP$ with $\ran (g) = \ran (f)$.
\end{itemize}

\subparagraph*{Relativized notions.} We relativize the concept of Turing machines and Turing transducers by giving them access to a write-only oracle tape. We relativize complexity classes, proof systems, machines, and (p-)simulation by defining them over machines with oracle access, i.e., whenever a Turing machine or Turing transducer is part of a definition, we replace them by an oracle Turing machine or an oracle Turing transducer.
For space-bounded Turing machines we utilize the model presented by Simon \cite{sim77}, in which the oracle tape also underlies the polynomial space-bound restriction. 
We indicate the access to some oracle $O$ in the superscript of the mentioned concepts, i.e., $\mathcal{C}^O$ for a complexity class $\mathcal{C}$ and $M^O$ for a Turing machine or Turing transducer $M$. We sometimes omit the oracles in the superscripts, e.g., when sketching ideas in order to convey intuition, but never in actual proofs. We also transfer all notations to the respective oracle concepts.

\section{The SPARSE-Relativization Framework}\label{sec:sparse-relativization}
\subparagraph{Bounded relativization.}Recently, Hirahara, Lu, and Ren \cite{hlr23} introduced the concept of \emph{bounded relativization}, which provides a framework to obtain proof barriers stronger than usual relativization.
\begin{definition}(Bounded relativization \cite{hlr23})\label{def:bounded-relativization}
    For a complexity class $\mathcal{C}$, a complexity-theoretic statement is $\mathcal{C}$-relativizing if it is true relative to every oracle $O \in \mathcal{C}$.
\end{definition}
Hirahara, Lu, and Ren raise the hypothesis that ``current proof techniques'' are $\PSPACE$-relativizing, i.e., the statements provable by ``current techniques'' hold relative to all oracles $O \in \PSPACE$. They underpin their hypothesis by showing that typical results proven by non-relativizing techniques like arithmetization are $\PSPACE$-relativizing, e.g., $\IP = \PSPACE$. Furthermore, they explain how this hypothesis provides strong proof barriers:
\begin{itemize}
    \item If some statement $S$ fails relative to an oracle $O \in \PSPACE$, it follows that ``current techniques'' (which are $\PSPACE$-relativizing by their hypothesis) can not prove $S$. 
    \item Furthermore, if the statement $S$ fails relative to an oracle $O \in \EXPH$, then ``current techniques'' can not prove $S$ unless the proof implies $\PSPACE \neq \EXPH$, which would be a breakthrough result in complexity theory. This follows from their hypothesis, since a proof for $S$ based on ``current techniques'' works relative to all oracles in $\PSPACE$, but it does not work relative to all oracles in $\EXPH$, thus separating $\PSPACE$ from $\EXPH$.  
\end{itemize}

\subparagraph{SPARSE-relativization framework.} We continue this line of research by proposing the notion of sparse-oracle-independence, from which we also obtain proof barriers stronger than usual relativization. 
\begin{definition}\label{def:sparse-oracle-independent}
    A complexity theoretic statement $S$ is sparse-oracle-independent if $S$ being true is equivalent to $S$ being true relative to every oracle from $\SPARSE$.
\end{definition}
Intuitively, we do not know whether a sparse-oracle-independent statement is true, but if it is true, then the statement is $\SPARSE$-relativizing in the sense of \cite{hlr23}. 
So the correctness of the statement is independent of whether it has access to a sparse oracle if it is a true statement.
We will first explain the usefulness of sparse-oracle-indepedent statements using an example, and then generalize it. 

Long and Selman \cite{ls86} and independently Balcázar, Book, and Schöning \cite{bbs86} show the following generalization of the Karp-Lipton theorem \cite{kl80}.
\begin{theorem}[\cite{bbs86,ls86}]\label{thm:bbs86}
    \begin{romanenumerate}
        \item If there exists a sparse set $S$ such that the polynomial-time hierarchy relative
to $S$ collapses, then the (unrelativized) polynomial-time hierarchy collapses.
        \item If there exists a sparse set $S$ such that the polynomial-time hierarchy relative to $S$ is infinite, then the (unrelativized) polynomial-time hierarchy is infinite.
    \end{romanenumerate}
\end{theorem}
In other words, we do not know whether the polynomial-time hierarchy is infinite or collapses, but we do know that either case is $\SPARSE$-relativizing. So both ``$\PH$ is infinite'' and ``$\PH$ collapses'' are sparse-oracle-independent statements. 

As Bovet, Crescenzi, and Silvestri \cite{bcs95} point out, an analysis of the corresponding proofs shows that Theorem \ref{thm:bbs86} itself relativizes, which leads to the following result.
\begin{theorem}\label{thm:ph-trick}
Let $A \subseteq \Sigma^*$ and $B \in \SPARSE$. $\PH^A$ collapses if and only if $\PH^{A \oplus B}$ collapses. 
\end{theorem}
This shows that combining an arbitrary oracle with a sparse oracle will not change the (in-)finiteness of the polynomial-time hierarchy relative to the arbitrary oracle. By Yao \cite{yao85} and Baker, Gill, and Solovay \cite{bgs75}, there is one oracle $A_{\text{Yao}}$ relative to which the polynomial-time hierarchy is infinite and one oracle $A_{\text{BGS}}$ relative to which it collapses. 

Assume that we want to prove the hardness of some statement $S$. A traditional oracle against $S$ tells us that
\begin{center}
    ``the statement $S$ is not relativizably provable''.
\end{center} 
However, by Theorem \ref{thm:ph-trick}, simply solving the task
\begin{center}
    \textbf{Task} constructSPARSE(statement $S$, set $A$):\\
    Construct $B \in \SPARSE$ such that $S$ fails relative to $A \oplus B$
\end{center}
for $S$ and $A_{\text{Yao}}$ immediately provides the stronger proof barrier
\begin{center}
    ``even when assuming an infinite $\PH$, the statement $S$ is not relativizably provable''.
\end{center}
Since an infinite polynomial-time hierarchy is a standard assumption in structural complexity theory, 
ruling out relativizable proofs for $S$ under this assumption is further evidence for the hardness of proving $S$. 

We can generalize this approach to arbitrary statements that are sparse-oracle-independent relative to all oracles, i.e., sparse-oracle-independent statements that if they are themselves relativized by an oracle, are still sparse-oracle-independent\footnote{Since the proofs that show ``a statement is sparse-oracle-independent'' usually relativize themselves, a sparse-oracle-independent statement usually holds relative to all oracles.}. 
\begin{theorem}\label{thm:sparse-relativization}
    Let $S_1$ be a statement that is sparse-oracle-independent relative to all oracles, let $A$ be an oracle relative to which $S_1$ holds, and let $B$ be according to \emph{constructSPARSE($S_2,A$)} for some statement $S_2$. Then $S_1$ holds and $S_2$ fails relative to $A \oplus B$. 
\end{theorem}
\begin{proof}
    The oracle $B$ is constructed such that $B \in \SPARSE$ and $S_2$ fails relative to $A \oplus B$. Furthermore, $S_1$ holds relative to $A$ and is sparse-oracle-independent relative to all oracles, hence $S_2$ also holds relative to $A \oplus B$.
\end{proof}

The requirement of constructSPARSE($S,A$) that ``$S$ fails relative to $A \oplus B$'' instead of $S$ just failing relative to $B$ is necessary for Theorem \ref{thm:sparse-relativization} as otherwise we get a contradiction. Either $\PH^\emptyset$ is infinite or collapsing. Then Theorem \ref{thm:sparse-relativization} would give that either $\PH^{\emptyset \oplus A_{\text{BGS}}}$ or $\PH^{\emptyset \oplus A_{\text{Yao}}}$ is both infinite and collapsing. 

However, in many cases it is not necessarily harder to achieve that ``$S$ fails relative to $A \oplus B$'' compared to ``$S$ fails relative to $B$''. This holds at least for oracle constructions that use diagonalization arguments that are independent of whether some base oracle $A$ already exists at the beginning of the construction. In fact, our oracle construction in Section~\ref{sec:oracle-construction} is of this type. This lets us describe a single oracle construction that works for all sets $A \subseteq \Sigma^*$ as base oracle. So, in these cases, constructing a sparse oracle relative to which $S$ fails lets us apply Theorem \ref{thm:sparse-relativization} for many different sparse-oracle-independent statements.

Besides Theorem \ref{thm:bbs86}, there are further sparse-oracle-independent statements. 
The following theorem presents some of them.

\begin{theorem}[sparse-oracle-independent statements by \cite{bbs86,ls86,kstt92}]\label{thm:sparse-relativizing-results}
    \ 
    \begin{enumerate}
        \item $\PH = \PSPACE$ if and only if $\PH^S = \PSPACE^S$ for all $S \in \SPARSE$. \cite[Prop.~3.13, Thm.~3.3]{ls86, bbs86}\label{thm:sparse-relativizing-results-1}
        \item $\PH \neq \PSPACE$ if and only if $\PH^S \neq \PSPACE^S$ for all $S \in \SPARSE$. \cite[Thm.~4.8]{bbs86}\label{thm:sparse-relativizing-results-2}
        \item $\PP \subseteq \PH$ if and only if $\PP^S \subseteq \PH^S$ for all $S \in \SPARSE$. \cite[Thm.~3.4]{bbs86}\label{thm:sparse-relativizing-results-3}
        \item $\PP \not \subseteq \PH$ if and only if $\PP^S \not \subseteq \PH^S$ for all $S \in \SPARSE$. \cite[Thm.~4.6]{bbs86}\label{thm:sparse-relativizing-results-4}
        \item $\NP \subseteq \CeqP$ if and only if $\NP^S \subseteq \CeqP^S$ for all $S \in \SPARSE$. \cite[Thm.~6.3]{kstt92}\label{thm:sparse-relativizing-results-5}
        \item $\NP \subseteq \modP$ if and only if $\NP^S \subseteq \modP^S$ for all $S \in \SPARSE$. \cite[Thm.~6.3]{kstt92}\label{thm:sparse-relativizing-results-6}
        \item $\modP \subseteq \PP$ if and only if $\modP^S \subseteq \PP^S$ for all $S \in \SPARSE$. \cite[Thm.~6.3]{kstt92}\label{thm:sparse-relativizing-results-7}
    \end{enumerate}
\end{theorem}

Bovet, Cresczeni, and Silvestri \cite[Thm.~4.3]{bcs95} generalize these types of sparse-oracle-independent statements by providing sufficent conditions for a statement to be sparse-oracle-independent.
Furthermore, the properties of the statements \ref{thm:sparse-relativizing-results-2} and \ref{thm:sparse-relativizing-results-4} of Theorem \ref{thm:sparse-relativizing-results} are true relative to $A_{\text{Yao}}$\footnote{Relative to this oracle, $\PSPACE$ and $\PP$ have complete sets, contrary to $\PH$.} and the properties of the remaining statements are true relative to $A_{\text{BGS}}$. 

In total, by Theorem \ref{thm:sparse-relativization}, for any statement $S_2$ and any sparse-oracle-independent statement $S_1$ of Theorems \ref{thm:bbs86} and \ref{thm:sparse-relativizing-results}, by constructing an oracle $B$ according to constructSPARSE($S_2,A$) with $A \in \{A_{\text{Yao}},A_{\text{BGS}}\}$, we obtain the barrier
\begin{center}
    ``even when assuming $S_1$, the statement $S_2$ is not relativizably provable''.
\end{center}
In Section \ref{sec:oracle-construction} we apply this approach to the statement ``all sets in $\PSPACE$ have p-optimal proof systems and there are arbitrarily complex sets with p-optimal proof systems''.

\subparagraph{TALLY-relativization framework.}
By augmenting Definition \ref{def:sparse-oracle-independent} using the class $\TALLY$ instead of $\SPARSE$, we say that a complexity theoretic statement $S$ is tally-oracle-independent if $S$ being true is equivalent to $S$ being true relative to every oracle from $\TALLY$. 
Using the same reasoning as for the $\SPARSE$-relativization framework, we can derive the $\TALLY$ versions of constructSPARSE (called constructTALLY) and Theorem \ref{thm:sparse-relativization}. Let us present some known tally-oracle-independent statements from the literature.
\begin{theorem}[Tally-oracle-independent statements by \cite{ls86}]\label{thm:tally-relativizing-results}
    \ 
    \begin{enumerate}
        \item $\P = \NP$ if and only if $\P^T = \NP^T$ for all $T \in \TALLY$. \cite[Cor.~3.5]{ls86}\label{thm:tally-relativizing-results-1}
        \item $\NP = \coNP$ if and only if $\NP^T = \coNP^T$ for all $T \in \TALLY$. \cite[Cor.~3.3]{ls86}\label{thm:tally-relativizing-results-2}
        \item $\P = \PSPACE$ if and only if $\P^T = \PSPACE^T$ for all $T \in \TALLY$. \cite[Prop.~3.11]{ls86}\label{thm:tally-relativizing-results-3}
        \item $\NP = \PSPACE$ if and only if $\NP^T = \PSPACE^T$ for all $T \in \TALLY$. \cite[Prop.~3.12]{ls86}\label{thm:tally-relativizing-results-4}
        \item $\NP \subseteq \E$ if and only if $\NP^T \subseteq \E^T$ for all $T \in \TALLY$. \cite[Prop.~3.10]{ls86}\label{thm:tally-relativizing-results-5}
    \end{enumerate}
\end{theorem}

Since all statements of Theorem \ref{thm:tally-relativizing-results} hold relative to the oracle $A_{\text{BGS}}$, we can derive the following application of above tally-oracle-independent statements. For any statement $S_2$ and any tally-oracle-independent statement $S_1$ of Theorem \ref{thm:tally-relativizing-results}, by constructing an oracle $B$ according to constructTALLY($S_2,A_{\text{BGS}}$), we obtain the barrier
\begin{center}
    ``even when assuming $S_1$, the statement $S_2$ is not relativizably provable''.
\end{center}

\subparagraph{Oracles from SPARSE$\cap$EXPH.} It is well-known that ``relativizable proof techniques'' do not capture our ``current techniques'' anymore. Indeed, $\SPARSE$-relativizing techniques also do not capture ``current techniques'' like arithmetization. By Shamir \cite{sha92} it holds that $\IP = \PSPACE$, but Fortnow and Sipser \cite{fs88} show the existence of a sparse oracle\footnote{The oracle constructed by Fortnow and Sipser is not sparse, but is easily transformed into a sparse oracle.} $A$ such that $\IP^A \neq \PSPACE^A$.
For this reason it is interesting to combine the bounded relativization approaches of Hirahara, Lu, and Ren \cite{hlr23} with ours.

Let $S_1$ be a sparse-oracle-independent statement that is true relative to an $\EXPH$-computable oracle $A$. For a statement $S_2$, let $B$ be constructed according to a $\SPARSE \cap \EXPH$ version of constructSPARSE($S_2,A$). If ``current techniques'' are $\PSPACE$-relativizing and can not show $\PSPACE \neq \EXPH$, we obtain the proof barrier 
\begin{center}
    ``even when assuming $S_1$, `current techniques' can not prove the statement $S_2$'',
\end{center}
because $A \oplus B \in \EXPH$ and thus \cite{hlr23} applies and also Theorem \ref{thm:sparse-relativization} for $S_1$, $S_2$, $A$, and $B$ applies. So oracles of this type provide the strongest proof barrier for statements, which motivates the following two tasks.
\begin{enumerate}
    \item Show the existence of $\EXPH$-computable oracles relative to which sparse-oracle-independent statements are true.
    \item Construct oracles according to a $\SPARSE \cap \EXPH$ version of constructSPARSE for statements\footnote{As mentioned before, this is possibly only as hard as constructing an oracle from $\SPARSE \cap \EXPH$ against statements.}.
\end{enumerate}
In particular, we are interested in the following question.
\begin{spreadlines}{0ex}
\begin{align*}
&\text{\textbf{Open}: Is there an $\EXPH$-computable oracle relative to which the }\\
&\text{\phantom{\textbf{Open}: }polynomial-time hierarchy is infinite?}
\end{align*}
\end{spreadlines}
To the best of our knowledge, we only know that the polynomial-time hierarchy is infinite relative to random oracles by H{\aa}stad, Rossman, Servedio, and Tan \cite{hrst17} and to some kind of space-random oracles by Hitchcock, Sekoni, and Shafei \cite[Thm.~4.4]{hss21}.

\subparagraph{Application of the SPARSE-relativization framework.}
Egidy and Glaßer \cite{eg25} prove the following theorem.
\begin{theorem}[\cite{eg25}]\label{thm:eg25}
    Let $A \subseteq 2\N$ be arbitrary. There is a sparse oracle $B \subseteq 2\N+1$ such that all $L \in \PSPACE^{B \cup A} \setminus \NP^{B \cup A}$ have no optimal proof system relative to $B \cup A$.
\end{theorem}
Furthermore, analyzing the proof of Theorem 4.5 of Balcázar, Book, and Schöning \cite{bbs86} gives the following result that also holds relative to all oracles.
\begin{corollary}[\cite{bbs86}, Corollary from Thm.~4.5]\label{cor:bbs-4.5-Cor}
    For all $k \geq 3$ it holds that
    \[\Sigma_k^{p} \neq \Pi_k^{p} \Longrightarrow \Sigma_{k-2}^{p,S} \neq \Pi_{k-2}^{p,S} \text{ for all } S \in \SPARSE.\]
\end{corollary}
Note that above corollary is different from Theorems \ref{thm:bbs86} and \ref{thm:sparse-relativizing-results}. The statement $\Sigma_k^{p} \neq \Pi_k^{p}$ may not be sparse-oracle-independent, but it is almost, since a sparse oracle preserves the separation of level $k$ of the polynomial-time hierarchy at least to level $k-2$.

Using our $\SPARSE$-relativization framework, we obtain the following theorem.
\begin{theorem}\label{thm:conn-and-ph}
    There is an oracle $O_1$ such that $\NP^{O_1} \subsetneq \PH^{O_1} \subsetneq \PSPACE^{O_1}$, $\PH^{O_1}$ collapses and all $L \in \PSPACE^{O_1} \setminus \NP^{O_1}$ have no optimal proof system relative to $O_1$.
\end{theorem}
\begin{proof}
    Let $S$ be the statement ``all $L \in \PSPACE \setminus \NP$ have no optimal proof system''. By Theorems \ref{thm:ph-trick} and \ref{thm:sparse-relativizing-results}, the statements ``$\PH$ collapses'' and ``$\PSPACE \neq \PH$'' are sparse-oracle-independent relative to all oracles. By Ko \cite[Thm.~4.4]{ko89}, there is an oracle $A_{\text{Ko}}$ relative to which $\PSPACE^{A_\text{Ko}} \neq \Sigma_4^{p,A_{\text{Ko}}} = \Pi_4^{p,A_{\text{Ko}}} \neq \Sigma_3^{p,A_{\text{Ko}}}$. By Theorem \ref{thm:eg25}, there is an oracle $B$ constructed according to constructSPARSE($\neg S,A_{\text{Ko}}$). Then Theorem \ref{thm:sparse-relativization} gives that there exists an oracle $O_1 \coloneqq A_{\text{Ko}} \oplus B$ such that $\PSPACE^{O_1} \neq \PH^{O_1}$, $\PH^{O_1}$ collapses, and all $L \in \PSPACE^{O_1} \setminus \NP^{O_1}$ have no optimal proof system relative to $O_1$.

    It remains to show that $\NP^{O_1} \neq \PH^{O_1}$. By Corollary \ref{cor:bbs-4.5-Cor} and $\Sigma_3^{p,A_{\text{Ko}}} \neq \Pi_3^{p,A_{\text{Ko}}}$ and $B \in \SPARSE$, it holds that $\NP^{O_1} = \Sigma_1^{p,O_1} \neq \Pi_1^{p,O_1} = \coNP^{O_1}$. Hence, $\NP^{O_1} \neq \PH^{O_1}$.
\end{proof}
We want to point out that constructing such an oracle by hand would be a particular challenge compared to the simple application of the $\SPARSE$-relativization framework. Theorem \ref{thm:conn-and-ph} gives the following insight on Q1 and Q2, because $\TAUT^{O_1}, \QBF^{O_1} \in \PSPACE^{O_1} \setminus \NP^{O_1}$.
\begin{corollary}
    A negative answer to \emph{Q1} and \emph{Q2} does not relativizably imply that $\PH$ is infinite.
\end{corollary}

\section{P-Optimal Proof Systems and Tally Oracles}\label{sec:oracle-construction}
In this section, we construct a tally oracle relative to which all sets in $\PSPACE$ have p-optimal proof systems and there are arbitrarily complex sets with p-optimal proof systems. By using an arbitrary oracle $A \subseteq \Sigma^*$ as base oracle at the beginning of the construction, we can apply Theorem \ref{thm:sparse-relativization} to obtain oracles relative to which additionally the tally-oracle-independent and sparse-oracle-independent statements of Theorems \ref{thm:bbs86}, \ref{thm:sparse-relativizing-results}, \ref{thm:tally-relativizing-results} hold.

\subsection{Notation for the Oracle Construction}
\subparagraph*{Words and sets.}
We denote the length of a word $w \in \Sigma^*$ by $|w|$. The empty word has length $0$ and is denoted by $\varepsilon$. The $(i+1)$-th letter of a word $w$ for $0 \leq i < |w|$ is denoted by $w(i)$, i.e., $w=w(0)w(1)\cdots w(|w|-1)$. If $v$ is a (strict) prefix of $w$, we write $v \sqsubseteq w$ ($v \sqsubsetneq w$) or $w \sqsupseteq v$ ($w \sqsupsetneq v$). 

We identify $\Sigma^*$ with $\N$ through a polynomial-time computable and polynomial-time invertible bijection $\mathrm{enc} \colon \Sigma^* \to \N$ defined as $\mathrm{enc}(w) = \sum _{i<|w|}(1+w(i))2^{|w|-1-i}$. Thus, we can treat words from $\Sigma^*$ as numbers from $\N$ and vice versa, which allows us to use notations, relations and operations of words for numbers and vice versa (e.g., we can define the length of a number by this). Expressions like $0^i$ and $1^i$ are ambiguous, because they can be interpreted as words from $\Sigma ^*$ or numbers. We resolve this ambiguity by the context or explicitly naming the interpretation. So a word $w \in \Sigma ^*$ can be interpreted as a set $\{i \in \N \mid w(i) = 1\}$
and subsequently as a partial oracle, which is defined for all words up to $|w|-1$. So $|w|$ is the first word that $w$ is not defined for, i.e., $|w| \in w1$ and $|w| \notin w0$. During the oracle construction we often use words from $\Sigma ^*$ to denote partially defined oracles. In particular, oracle queries for undefined words of a partial oracle are answered negatively.

\subparagraph*{Functions.}
A function $f'$ is an extension of a function $f$, denoted by $f \sqsubseteq f'$ or $f' \sqsupseteq f$, if $\dom(f) \subseteq \dom(f')$ and $f(x) = f'(x)$ for all $x \in \dom (f)$. If $x \notin \dom (f)$, then $f \cup \{x \mapsto y\}$ denotes the extension $f'$ of $f$ such that $f'(z) = f(z)$ for $z \not = x$ and $f'(x)=y$. If $y \in \ran (f)$ for some proof system $f$, then we denote the shortest $f$-proof for $y$ by $y_f$, i.e., $f(y_f)=y$ and for all $x < y_f$ it holds that $f(x) \not = y$. If $y \notin \ran (f)$, then $y_f$ is not defined.

We define polynomial functions $p_k \colon \N \to \N$ for $k \in \N^+$ by $p_k(n) \coloneqq n^k + k$. Also, let $\langle \cdot \rangle \colon \bigcup _{i \geq 0} \N^i \to \N$ be an injective polynomial-time computable and polynomial-time invertible sequence encoding function such that $|\langle u_1, \dots , u_n \rangle | = 2(|u_1| + \cdots + |u_n| + n)$.

\subparagraph{Definite computations}
A computation $F^w(x)$ is called definite, if $w$ is defined for all words of length less than or equal to the runtime bound of $F^w(x)$.
Additionally, a space-bounded computation $M^w(x)$ with bound $p(|x|)$ is definite, if $w$ is defined for all words of length less than or equal to $p(|x|)$. A computation definitely accepts (resp., rejects), if it is definite and accepts (resp., rejects). We may combine computations to more complex expressions and call those definite, if every individual computation is definite, e.g., ``$F^w(x) \in L(M^w)$ is definite'' means $F^w(x)$ outputs some $y$ and is definite, and $M^w(y)$ is definite.

\subparagraph*{Enumerations.} For any oracle $O$ let $\{F_k^O\}_{k \in \N^+}$ be a standard enumeration of polynomial-time oracle Turing transducers, where $F_k^O$ has running time exactly $p_k$. For any oracle $O$ let $\{S_k^O\}_{k \in \N^+}$ denote a standard enumeration of polynomial space oracle Turing machines, where $S_k^O$ has space-bound exactly $p_k$. Let $t_1,t_2,\dots$ be a not necessarily computable enumeration of all time-constructible functions.

\subsection{Oracle Construction}\label{sec:oracle-construction-2}
For Section \ref{sec:oracle-construction-2}, let $A \subseteq 2\N$ be arbitrary. Note that by the definition of $\mathrm{enc}$ it holds that $A \cap 0^* = \emptyset$.

\subparagraph*{Relativized PSPACE-complete set.}
\begin{definition}[$\PSPACE^O$-complete set]
    For all oracles $O$, define
    \begin{align*}
        K^O \coloneqq \{\langle 0^i,0^{t},x\rangle \mid i,t \in \N^+, x \in \Sigma^* \text{ and } S_i^O(x) \text{ accepts using space} \leq t\}.
    \end{align*}
    Furthermore, fix $\kappa$ as the number of the first machine $S_\kappa$ such that $L(S_\kappa^O) = K^O$.
\end{definition}
\begin{observation}
    For all oracles $O$, $K^O$ is $\PSPACE^O$-complete.
\end{observation}
Note that any computation $S_\kappa^O(\langle 0^i, 0^t, x \rangle)$ only depends on oracle words of length at most $t$.

\subparagraph*{Code words for PSPACE-complete set.} For machines $F_k$ and $n \in \N$, we want to encode facts of the form ``$F_k^O$ is a sound proof system for $K^O$ on inputs of length $n$''. We pad these encodings such that they can not be queried by some computations. 

\begin{definition}[Code words for $K^O$]\label{def:cK} Define 
    \[c_K \colon \N^+ \times \N \to \Sigma^*,\quad c_K(k,n) \coloneqq 0^w \text{ where }w = {\langle 0,k,n,p_{k}(n) \rangle}.\]
\end{definition}
\begin{observation}\label{obs:cK-polytime}
    For fixed $k$, the function $n \mapsto c_K(k,n)$ is polynomial-time computable with respect to $n$.
\end{observation}
\begin{observation}\label{obs:cK-definite}
    For $k \in \N^+$ and $O \subseteq \Sigma^*$, $S_\kappa^O(F_k^O(x))$ does not depend on words of length $\geq |c_K(k,|x|)|$.
\end{observation}
\subparagraph*{Code words for arbitrary complex sets.} Relative to our oracle, we want to have p-optimal proof systems for arbitrary complex sets. We define these sets by using a second type of code word.
\begin{definition}[Code words for arbitrary complex sets]\label{def:cA} Define 
    \[c_L \colon \N^+ \times \N \to \Sigma^*, \quad c_L(i,n) \coloneqq 0^w \text{ where } w = {\langle 1,i,n,t_i(n) \rangle}.\]
\end{definition} 
\begin{definition}[Arbitrary complex sets]\label{def:Ai} Let $i \in \N^+$ and $O$ be an arbitrary oracle. We define
    \[L_i^O \coloneqq \{0^n \mid c_L(i,n) \in O\}.\]
    For a partial oracle $w$, we say that $x \in L_i^{w \cup A}$ or $x \notin L_i^{w \cup A}$ is definite if $w$ is defined for all words of length $\leq |c_L(i,|x|)|$.
\end{definition}
Intuitively, the complexity of $L_i^O$ depends on $t_i$. We construct the oracle such that this holds true.
\begin{observation}\label{obs:cA-polytime}
    For fixed $i$, the function $n \mapsto c_L(i,n)$ is polynomial-time computable with respect to $|c_L(i,n)|$, because $t_i$ is time-constructible.
\end{observation}
\begin{observation}\label{obs:cK-cA-properties}
    It holds that $c_K$ and $c_L$ are total, injective, polynomial-time invertible, $\ran(c_K) \cap \ran(c_L) = \emptyset$, $\ran(c_K) \cup \ran(c_L) \subseteq 0^*$ and $\ran(c_K) \in \P$.
\end{observation}
\subparagraph*{Valid oracles.}
During the construction of the oracle, we successively add requirements that we maintain. These are specified by a partial function $r \colon \N^+ \times \N^+ \mapsto \N$, called requirement function. An oracle $w$ is $r$-valid for a requirement function $r$, if it satisfies the following requirements.
\begin{enumerate}[V1]
    \item\label{V1} $w \subseteq \ran(c_K) \cup \ran(c_L)$.

    (Meaning: The oracle is tally and contains only code words.)
    \item\label{V2} If $c_K(k,n) < |w|$, then $c_K(k,n) \in w$ if and only if for all $x \in \Sigma^n$, $F_k^{w \cup A}(x) \in K^{w \cup A}$.

    (Meaning: The oracle indicates the soundness of proof systems for $K^{w \cup A}$.)

    \item\label{V3} If $r(k,i) = 0$, then there is some $x$ such that $F_k^{w \cup A}(x) \notin L_i^{w \cup A}$ is definite.

    (Meaning: If $r(k,i)=0$, then $F_k$ is no proof system for $L_i$ relative to $w \cup A$ and all its extensions.)
    \item\label{V4} If $r(k,i) = m > 0$, then for all $x \in L_i^{w \cup A}$ with $|x| > m$ it holds that $x_{F_k^{w \cup A}}$ is not defined or $p_k(|x_{F_k^{w \cup A}}|) \geq |c_L(i,|x|)|$.

    (Meaning: If $r(k,i)>0$, then $F_k$ does not have short proofs for elements from $L_i$.)
    \item\label{V5} For all $i \in \N^+$ it holds that $\card{L_i^{w \cup A}} \geq \card{\{k \in \N^+ \mid (k,i) \in \dom(r)\}}$. 
    
    (Meaning: The sets $L_i$ contain a minimum amount of elements.)
\end{enumerate}
We will prove that V\ref{V1} to V\ref{V5} are satisfied for various oracles and requirement functions. To prevent confusion, we use the notation V\ref{V1}($w$), V\ref{V2}($w$), V\ref{V3}($w,r$), V\ref{V4}($w,r$), and V\ref{V5}($w,r$) to clearly state the referred oracle $w$ and requirement function $r$.
\subparagraph{Properties of $r$-valid oracles}
The following lemma shows how to extend $r$-valid oracles such that they remain $r$-valid.
\begin{lemma}\label{lem:extension}
    If $w$ is an $r$-valid oracle, then there exists $b \in \{0,1\}$ such that $wb$ is $r$-valid.
\end{lemma}
\begin{proof}
    Let $w$ be an arbitrary $r$-valid oracle and let $z \coloneqq |w|$ be the word whose membership is decided next. We choose $b$ according to the following procedure:
    \begin{romanenumerate}
        \item\label{lem:extension-i} If $z = c_K(k,n)$ for some $k \in \N^+$ and $n \in \N$ and for all $x \in \Sigma^n$ it holds that $F_k^{w \cup A} (x) \in K^{w \cup A}$, then $b \coloneqq 1$.
        \item\label{lem:extension-ii} Otherwise, $b \coloneqq 0$. 
    \end{romanenumerate}
    We show that $wb$ is $r$-valid.

    \emph{To V\ref{V1}($wb$)}: Since V\ref{V1}($w$) is satisfied, it holds that $w \subseteq \ran(c_K) \cup \ran(c_L)$. Also, $b$ is only chosen as $1$ when $z = c_K(k,n)$ for some $k$ and $n$. Thus, $wb \subseteq \ran(c_K) \cup \ran(c_L)$.
    
    \emph{To V\ref{V2}($wb$)}: V\ref{V2}($w$) is satisfied. By Observation \ref{obs:cK-definite}, for any $c_K(k,n) < |w|$ it holds that for all $x \in \Sigma^n$, $S_{\kappa}^{w \cup A}(F_k^{w \cup A}(x))$ only depends on words of length $<|c_K(k,n)|$. So V\ref{V2}($wb$) remains satisfied for such code words $c_K(k,n)$. 
    
    If $z = c_K(k,n)$ for some $k$ and $n$, then by the injectivity of $c_K$ (Observation \ref{obs:cK-cA-properties}), $k$ and $n$ are unique. Furthermore, again by Observation \ref{obs:cK-definite}, for all $x \in \Sigma^n$ and any $c \in \{0,1\}$ it holds that
    \[S_\kappa^{w \cup A}(F_k^{w \cup A}(x)) = S_\kappa^{wc \cup A}(F_k^{wc \cup A}(x)).\]
    Thus, case (\ref{lem:extension-i}) extends $w$ by adding $c_K(k,n)$ precisely if and only if it is necessary for V\ref{V2}($wb$).
   
    \emph{To V\ref{V3}($wb,r$)}: Since V\ref{V3}($w,r$) is satisfied and the respective computations are definite, V\ref{V3}($wb,r$) is satisfied for any choice of $b$.
    
    \emph{To V\ref{V4}($wb,r$)}: Let $r(k,i) = m > 0$ and let $x \in L_i^{w \cup A}$ with $|x| > m$ be arbitrary. By the definition of $L_i$, if $x \in L_i^{w \cup A}$, then $c_L(i,|x|) \in w$. Thus $||w|| \geq |c_L(i,|x|)|$ and $F_k^{w \cup A}(y)$ is definite for all $y$ such that $p_k(|y|) < |c_L(i,|x|)|$. By V\ref{V4}($w,r$), it holds that $x_{F_k^{w \cup A}}$ is not defined or $p_k(|x_{F_k^{w \cup A}}|) \geq |c_L(i,|x|)|$. Hence, $F_k^{wb \cup A}(y) = F_k^{w \cup A}(y) \neq x$ for all $y$ such that $p_k(|y|) < |c_L(i,|x|)|$. Finally, either $x_{F_k^{wb \cup A}}$ is not defined or $p_k(|x_{F_k^{wb \cup A}}|) \geq |c_L(i,|x|)|$. This shows that V\ref{V4}($wb,r$) is satisfied for all words $x \in L_i^{w \cup A}$. By the disjointness of $\ran(c_K)$ and $\ran(c_L)$ (Observation \ref{obs:cK-cA-properties}) it holds that if $z = c_L(i,n)$ for some $i$ and $n$, then case (\ref{lem:extension-ii}) extends $w$ to $w0$. Consequently, $L_i^{wb \cup A} = L_i^{w \cup A}$ and V\ref{V4}($wb,r$) is also satisfied for all words in $x \in L_i^{wb \cup A}$. 

    \emph{To V\ref{V5}($wb,r$)}: Since V\ref{V5}($w,r$) is satisfied and $wb \sqsupseteq w$, also V\ref{V5}($wb,r$) is satisfied.
\end{proof}
The following lemma shows that if we can construct an $r$-valid oracle $O$ with $\dom(r) = \N^+ \times \N^+$, then $O$ has the desired properties.
\begin{lemma}\label{lem:properties-of-valid-oracle}
   Let $r$ be a requirement function with $\dom(r) = \N^+ \times \N^+$ and let $O$ be an $r$-valid oracle. Then the following holds:
   \begin{romanenumerate}
    \item\label{lem:properties-of-valid-oracle-i} $O \in \TALLY$. 
    \item\label{lem:properties-of-valid-oracle-ii} $K^{O \cup A}$ has p-optimal proof systems.
    \item\label{lem:properties-of-valid-oracle-iii} For all $i \in \N^+$, $L_i^{O \cup A}$ has p-optimal proof systems.
    \item\label{lem:properties-of-valid-oracle-iv} For any time-constructible function $t$ there is some $L_i^{O \cup A} \notin \DTIME^{O \cup A}(t)$.
    \end{romanenumerate} 
\end{lemma}
\begin{proof}
    \emph{To (\ref{lem:properties-of-valid-oracle-i})}: Follows by V\ref{V1}($O$) and Observation \ref{obs:cK-cA-properties}.
    \bigskip

    \centerline{\emph{To (\ref{lem:properties-of-valid-oracle-ii})}:}
    Let $f$ be any proof system for $K^{O \cup A}$. Define
    \[h(x) \coloneqq 
    \begin{cases}
        F_k^{O \cup A}(x') & \text{if } x = \langle x', c_K(k,|x'|) \rangle \text{ and } c_K(k,|x'|) \in O\\
        f(x') & \text{if } x = \langle x', 1 \rangle\\
        f(0) & \text{otherwise}
    \end{cases}\]
    \begin{claim}\label{claim:h-in-FP}
        $h \in \FP^{O \cup A}$.
    \end{claim}
    \begin{claimproof}
        The encoding $\langle \rangle$ is polynomial-time invertible, $\ran(c_K) \in \P$ and $c_K$ is polynomial-time invertible (Observation \ref{obs:cK-cA-properties}). If the second part of the input is a code word, we can query the oracle in linear time and compute the output in polynomial time, since in this case the input has length at least $p_k(|x'|)$. 
    \end{claimproof}
    \begin{claim}\label{claim:ranh-equals-KO}
        $\ran(h) = K^{O \cup A}$. 
    \end{claim}
    \begin{claimproof}
        Since the second line of the definition of $h$ ensures $\ran(h) \supseteq \ran(f)$ and $f$ is a proof system for $K^{O \cup A}$, $\ran(h) \supseteq K^{O \cup A}$.

        If the input $x$ does not meet the requirements of the first line of the definition of $h$, then the output of $h$ is an element from $K^{O \cup A}$. Otherwise, $x = \langle x', c_K(k,|x'|)\rangle$ and $c_K(k,|x'|) \in O$. By V\ref{V2}($O$), it follows that for all $z \in \Sigma^{|x'|}$, $F_k^{O \cup A}(z) \in K^{O \cup A}$. Hence, $F_k^{O \cup A}(x') \in K^{O \cup A}$. This shows $\ran(h) \subseteq K^{O \cup A}$.        
    \end{claimproof}
    \begin{claim}\label{claim:h-psimulates-all-ps}
        All proof systems for $K^{O \cup A}$ are p-simulated by $h$.
    \end{claim}
    \begin{claimproof}
        Let $g$ be an arbitrary proof system for $K^{O \cup A}$. Let $F_k^{O \cup A}$ be the machine computing $g$. Consider the following simulation function $\pi$:
        \[\pi (x) \coloneqq \langle x, c_K(k,|x|)\rangle\]
        By Observation \ref{obs:cK-polytime}, for fixed $k$, we can compute $c_K(k,|x|)$ in polynomial time in $|x|$ and thus, $\pi \in \FP$. Furthermore, it holds that
        \[h(\pi(x)) = h(\langle x, c_K(k,|x|)\rangle) = F_k^{O \cup A}(x) = g(x).\]
        We have to show that the second equality holds. By $\ran(g) = K^{O \cup A}$, for all $x \in \Sigma^*$, $F_k^{O \cup A}(x) \in K^{O \cup A}$. Together with V\ref{V2}($O$), this means that $c_K(k,n) \in O$ for all $n \in \N$. Hence, $c_K(k,|x|) \in O$ and the second equality holds.
    \end{claimproof}
    The Claims \ref{claim:h-in-FP}, \ref{claim:ranh-equals-KO}, and \ref{claim:h-psimulates-all-ps} show that $h$ is a p-optimal proof system for $K^{O \cup A}$.
    \bigskip

    \centerline{\emph{To (\ref{lem:properties-of-valid-oracle-iii})}:}
    Let $i \in \N^+$ be arbitrary and let $a$ be any element of $L_i^{O \cup A}$ (which exists by V\ref{V5}($O,r$) and $r$ being total). Define
    \[h(x) \coloneqq 
    \begin{cases}
        0^n & \text{if } x = c_L(i,n) \in O\\
        a & \text{otherwise}
    \end{cases}
    \]
    Since by V\ref{V1}($O$), $O \subseteq \ran(c_K) \cup \ran(c_L)$ and $\ran(c_K) \in \P$ (Observation \ref{obs:cK-cA-properties}), we can check for line 1 in polynomial time. Since we can compute $0^n$ from $c_L(i,n)$, we get that $h \in \FP^{O \cup A}$. Furthermore, $\ran(h) = L_i^{O \cup A}$, because $a \in L_i^{O \cup A}$ and the first line is according to the definition of $L_i^{O \cup A}$. So $h$ is a proof system for $L_i^{O \cup A}$. 

    It remains to show that $h$ is p-optimal. Let $g$ be an arbitrary proof system for $L_i^{O \cup A}$. Let $F_k^{O \cup A}$ be the machine computing $g$. Consider the following simulation function $\pi$:
    \[\pi(x) \coloneqq c_L(i,|F_k^{O \cup A}(x)|)\]
    Since $L_i^{O \cup A} \in \TALLY$ and $F_k^{O \cup A}$ is a proof system for $L_i^{O \cup A}$, it holds that for all $x$
    \[F_k^{O \cup A}(x) = 0^{|F_k^{O \cup A}(x)|} \in L_i^{O \cup A} \quad \text{and} \quad c_L(i,|F_k^{O \cup A}(x)|) \in O.\]
    Using this, we obtain for all $x$
    \[h(\pi(x)) = h(c_L(i,|F_k^{O \cup A}(x)|)) = 0^{|F_k^{O \cup A}(x)|} = F_k^{O \cup A}(x) = g(x).\]
    It remains to show that $\pi \in \FP^{O \cup A}$. By the prerequisite that $\dom(r) = \N^+ \times \N^+$, it holds that $(k,i) \in \dom(r)$. Moreover, $r(k,i) = m > 0$, because otherwise by V\ref{V3}($O,r$), there is some $x$ such that $F_k^{O \cup A}(x) \notin L_i^{O \cup A}$, a contradiction that $F_k^{O \cup A}$ is a proof system for $L_i^{O \cup A}$. By V\ref{V4}($O,r$), it holds that for $y \in L_i^{O \cup A}$ with $|y| > m$, $p_k(|y_{F_k^{O \cup A}}|) \geq |c_L(i,|y|)|$. We use this to bound the complexity of $\pi$.
    
    Let $x$ be arbitrary. If $|F_k^{O \cup A}(x)| \leq m$, then $c_L(i,|F_k^{O \cup A}(x)|)$ is bounded by some constant, because there are only finitely many such code words. Otherwise, let $y \coloneqq F_k^{O \cup A}(x)$ with $|y| > m$. Then it holds that 
    \[p_k(|x|) \geq p_k(|y_{F_k^{O \cup A}}|) \geq |c_L(i,|y|)|.\]
    So $c_L(i,|F_k^{O \cup A}(x)|)$ is at most polynomially longer than $|x|$. Since $x \mapsto |F_k^{O \cup A}(x)|$ can be computed in polynomial time in $|x|$ and $|F_k^{O \cup A}(x)| \mapsto c_L(i,|F_k^{O \cup A}(x)|)$ can be computed in polynomial time in $|c_L(i,|F_k^{O \cup A}(x)|)|$ (Observation \ref{obs:cA-polytime}), we can compute $x \mapsto c_L(i,|F_k^{O \cup A}(x)|)$ in polynomial time in $|x|$. Thus, $\pi \in \FP^{O \cup A}$ and $h$ is p-optimal.
    \bigskip

    \centerline{\emph{To (\ref{lem:properties-of-valid-oracle-iv}):}} First, we have to show that for all $i \in \N^+$, $L_i^{O \cup A}$ is not finite.
    \begin{claim}\label{claim:Ai-infinite}
       For all $i \in \N^+$ it holds that $\card{L_i^{O \cup A}} = \infty$.
    \end{claim}
    \begin{claimproof}
        Follows by V\ref{V5}($O,r$) and $\dom(r) = \N^+ \times \N^+$.
    \end{claimproof} 
    Let $t$ be an arbitrary time-constructible function and let $t' \coloneqq 2^t$. Then $t'$ is also time-constructible, so there is some $i$ such that $t' = t_i$.
    
    Suppose that $L_i^{O \cup A} \in \DTIME^{O \cup A}(t)$ via a Turing machine $M^{O \cup A}$. Let $a$ be any element of $L_i^{O \cup A}$ (which exists by Claim \ref{claim:Ai-infinite}). Consider the following proof system $g$:
    \[g(x) \coloneqq 
    \begin{cases}
        y & \text{if } x \text{ is an accepting computation path of } M^{O \cup A}(y)\\
        a & \text{otherwise}    
    \end{cases}
    \]
    Since by assumption $M^{O \cup A}$ decides $L_i^{O \cup A}$ in time $O(t)$, $g$ has proofs for the elements of $L_i^{O \cup A}$ of  length $\leq c \cdot t$ for some constant $c$. By (\ref{lem:properties-of-valid-oracle-iii}), $L_i^{O \cup A}$ has a p-optimal proof system $h$ that for almost all $x \in L_i^{O \cup A}$ has proofs of length $|c_L(i,|x|)| \geq t_i(|x|)$. Since $h$ is a p-optimal proof system, $h$ p-simulates $g$ via some $\pi \in \FP^{O \cup A}$. Let $p$ denote the runtime polynomial for $\pi$. Then $\pi$ can increase the length of proofs at most by $p$. This gives a contradiction to the p-optimality of $h$, because $\card{L_i^{O \cup A}} = \infty$ (Claim \ref{claim:Ai-infinite}), thus $h$ has infinitely often proofs of length $\geq t_i$. Since there is some $m$ such that for all $n \geq m$ it holds that 
    \[p(c \cdot t(n)) < 2^{t(n)} = t_i(n),\]
    $\pi$ can not translate all $g$-proofs into $h$-proofs. Hence, $L_i^{O \cup A} \notin \DTIME^{O \cup A}(t)$.
\end{proof}
The remaining part of this section deals with the construction of a total requirement function $r$ and an $r$-valid oracle.

\subparagraph{Oracle construction.}
The oracle construction will take care of the following set of tasks $\{\tau_{k,i} \mid k,i \in \N^+\}$. Let $T$ be an enumeration of these tasks. In each step we treat the smallest task in the order specified by $T$, and after treating a task we remove it from $T$.

In step $s = 0$ we define $w_0 \coloneqq \varepsilon$ and $r_0$ as the nowhere defined requirement function. In step $s > 0$ we define $w_s \sqsupsetneq w_{s-1}$ and $r_s \sqsupsetneq r_{s-1}$ such that $w_s$ is $r_s$-valid by treating the earliest task in $T$ and removing it from $T$. We do this according to the following procedure:

\begin{itemize}
\item[] {\bf Task} $\tau_{k,i}$: Let $r^0 \coloneqq r_{s-1} \cup \{(k,i) \mapsto 0\}$. If there exists an $r^0$-valid partial oracle $v \sqsupsetneq w_{s-1}$, then assign $r_s \coloneqq r^0$ and $w_s \coloneqq v$. 

Otherwise, $n \in \N$ such that $c_L(i,n)$ is the lexicographic smallest code word for which $w_{s-1}$ is not defined. Extend $w_{s-1}$ to $v$ via Lemma~\ref{lem:extension} such that $|v| = c_L(i,n)$. Let $w_s \coloneqq v1$ and $r_s \coloneqq r_{s-1} \cup \{(k,i) \mapsto ||w_{s}||+1\}$.

(Meaning: Try to extend the oracle such that $F_k$ is no proof system for $L_i$. If this is not possible, then enforce that $F_k$ does not have short proofs for elements from $L_i$. In either case, $L_i$ receives an additional element.)
\end{itemize}
This procedure yields two series $w_0, w_1, \dots$ and $r_0, r_1, \dots$, from which we define our desired oracle and its associated requirement function.
\begin{definition}[Desired Oracle]\label{def:oracle-definition}
    Define $O \coloneqq \bigcup_{s\in \N}w_s$ and $r \coloneqq \bigcup_{s\in \N}r_s$.   
\end{definition}
The following claims show that $O$ and $r$ are well-defined and that $O$ is $r$-valid. 

\begin{claim}\label{claim:O-r-well-defined}
    $O$ and $r$ are well-defined, $O$ is defined for all words from $\Sigma^*$ and $\dom(r) = \N^+ \times \N^+$.
\end{claim}
\begin{claimproof}
    For all $s \in \N$, we show that $w_{s+1} \sqsupsetneq w_{s}$, $r_{s+1} \sqsupsetneq r_{s}$, and $\dom(r) = \N^+ \times \N^+$, from which the claim follows.
    
    Let $s \in \N$ be arbitrary and $\tau_{k,i}$ be the task treated at step $s+1$. When treating any task at step $s+1$, $w_s$ gets extended by at least one bit and $r_s$ gets defined for $(k,i)$. The enumeration $T$ does not have duplicate tasks and when a task is treated, it gets removed from $T$ in the same step. Hence, $r_s$ can not already be defined for $(k,i)$. Furthermore, every task in $T$ gets treated once, thus $\dom(r) = \N^+ \times \N^+$.
\end{claimproof}

\begin{claim}\label{claim:ws-is-rs-valid}
    For every $s \in \N$ it holds that $w_s$ is $r_s$-valid.
\end{claim}
\begin{claimproof}
    Observe that V\ref{V1}($w_0$) and V\ref{V2}($w_0$) are satisfied and that V\ref{V3}($w_0,r_0$), V\ref{V4}($w_0,r_0$), and V\ref{V5}($w_0,r_0$) do not require anything.

    Let $s \in \N^+$ be arbitrary. By induction hypothesis, $w_{s-1}$ is $r_{s-1}$-valid. Let $\tau_{k,i}$ be the task treated at step $s$. 

    Either, $w_s$ was chosen to be $r_s$-valid. Otherwise, $w_{s-1}$ got extended to $v$ via Lemma \ref{lem:extension}, resulting in  $v$ being $r_{s-1}$-valid. 
    \begin{claim}
        $w_s$ is $r_{s-1}$-valid.
    \end{claim}
    \begin{claimproof}
    Since $|v| \in \ran(c_L)$, V\ref{V1}($w_s$) is satisfied. Using the same argument as in Lemma \ref{lem:extension} and that $\ran(c_L)$ is disjoint from $\ran(c_K)$, V\ref{V2}($w_s$) is also satisfied. As argued in Lemma \ref{lem:extension}, V\ref{V3}($w_s,r_{s-1}$) and V\ref{V5}($w_s,r_{s-1}$) remain satisfied, because the oracle is only getting extended.

    For V\ref{V4}($w_s,r_{s-1}$), we use a proof by contradiction. Suppose that V\ref{V4}($w_s,r_{s-1}$) is violated. Since V\ref{V4}($v,r_{s-1}$) is not violated, no $x \in L_{j}^{w_s \cup A}$ with $c_L(j,|x|) < |v|$ is a witness for the violation, since otherwise the respective shortest proof for $x$ would be long enough to depend on words of length $||v|| \geq |c_L(j,|x|)|$. This contradicts that this proof can not depend on $c_L(j,|x|)$ when violating V\ref{V4}($w_s,r_{s-1}$).
    
    Hence, only $0^n$ with $c_L(i,n) \coloneqq |v|$ can be such a witness. More precisely, there is some $\ell$ such that $r_{s-1}(\ell,i) = m > 0$, $n > m$, and $p_\ell(|0^n_{F_\ell^{w_s \cup A}}|) < |c_L(i,n)|$. We show that in step $s'<s$ in the oracle construction, when treating the task $\tau_{\ell,i}$, the construction would have extended $w_{s'-1}$ to $v0$ and $r_{s'-1}$ to $r^0 = r_{s'-1} \cup \{(\ell,i) \mapsto 0\}$, contradicting $r_{s-1}(\ell,i) > 0$. 
    
    By Lemma \ref{lem:extension}, $v0$ is $r_{s-1}$-valid and thus also $r_{s'-1}$-valid. Also, $c_L(i,n) \notin v0$, so $0^n \notin L_i^{v0 \cup A}$. By
    \[p_\ell(|0^n_{F_\ell^{w_s \cup A}}|) < |c_L(i,n)| = ||v|| \leq ||v0|| = ||w_s||,\]
    it holds that
    \[0^n_{F_\ell^{w_s \cup A}} = 0^n_{F_\ell^{v0 \cup A}}.\] 
    So there exists some $x$ (namely $0^n_{F_\ell^{v0 \cup A}}$) such that $F_\ell^{v0 \cup A}(x) = 0^n \notin L_i^{v0 \cup A}$ and this is definite. Consequently, V\ref{V3}($v0,r^0)$ is satisfied for $(\ell,i)$. Hence, $v0$ is $r^0$-valid. In step $s'$, the oracle construction would have chosen $v0 \sqsupsetneq w_{s'-1}$ and $r^0 = r_{s'-1} \cup \{(\ell,i) \mapsto 0\}$ instead of $w_{s'}$ and $r_{s'}$ with $r_{s'}(\ell,i) > 0$, a contradiction to $r_{s-1} \sqsupseteq r_{s'}$. This shows that also V\ref{V4}($w_s,r_{s-1}$) is satisfied.
    \end{claimproof}
    It remains to show that $w_s$ is even $r_s$-valid. Since $r_s(k,i) > ||w_s||$, we only have to argue for V\ref{V4}($w_s,r_s$) and V\ref{V5}($w_s,r_s$). Since $L_i^{w_s \cup A}$ contains no word of length $>||w_s||$, going from V\ref{V4}($w_s,r_{s-1}$) to V\ref{V4}($w_s,r_s$) does not make a difference. Since $w_s$ is chosen as $v1$ with $c_L(i,n) = |v|$, it holds that 
    \begin{align*}
        \card{L_i^{w_s \cup A}} &\geq \card{L_i^{w_{s-1} \cup A}}+1 \geq \card{\{k \in \N^+ \mid (k,i) \in \dom(r_{s-1})\}}+1\\
        &= \card{\{k \in \N^+ \mid (k,i) \in \dom(r_{s})\}}.
    \end{align*}
    Hence, V\ref{V5}($w_s,r_s$) is satisfied and thus $w_s$ is $r_s$-valid.   
\end{claimproof}

\begin{claim}\label{claim:O-r-valid}
    $O$ is $r$-valid.
\end{claim}
\begin{claimproof}
    Intuitively, any violation of V\ref{V1} to V\ref{V5} for $O$ and $r$ leads to a contradiction to Claim~\ref{claim:ws-is-rs-valid}, i.e., that $w_s$ is $r_s$-valid for all $s \in \N$.

    \emph{To V\ref{V1}($O$)}: Suppose there is some $a \notin \ran(c_K) \cup \ran(c_L)$ such that $a \in O$, then $a \in w_s$ for some $s \in \N$. Consequently, V\ref{V1}($w_s$) would not be satisfied, a contradiction to Claim \ref{claim:ws-is-rs-valid}.

    \emph{To V\ref{V2}($O$)}: Suppose there is $k$ and $n$ such that $c_K(k,n) \notin O$ if and only if for all $x \in \Sigma^n$, $F_k^{O \cup A}(x) \in K^{O \cup A}$, then consider some oracle $w_s$ such that $||w_s|| > |c_K(k,n)|$. By Observation~\ref{obs:cK-definite}, $S_\kappa^{O \cup A}(F_k^{O \cup A}(x))$ only query words of length $<|c_K(k,n)|$. Thus, $c_K(k,n) \notin w_s$ if and only if for all $x \in \Sigma^n$, $F_k^{w_s \cup A}(x) \in K^{w_s \cup A}$, contradicting that $w_s$ is $r_s$-valid. 

    \emph{To V\ref{V3}($O,r$)}: If $r(k,i) = 0$, let $s$ be the step of the oracle construction that treated the task $\tau_{k,i}$. Since $w_s$ is $r_s$-valid, there is some $x$ such that $F_k^{w_s \cup A}(x) \notin L_i^{w_s \cup A}$ is definite. Hence, $F_k^{O \cup A}(x) \notin L_i^{O \cup A}$, showing that V\ref{V3}($O,r$) is satisfied.
    
    \emph{To V\ref{V4}($O,r$)}: Suppose that $r(k,i) = m > 0$ and that there is some $x \in L_i^{O \cup A}$ with $|x| > m$ such that $p_k(|x_{F_k^{O \cup A}}|) < |c_L(i,|x|)|$. Consider the first step $s$ where $r_s(k,i)=m$ and $||w_s|| > |c_L(i,|x|)|$. Then $x \in L_i^{w_s \cup A}$ is definite. Furthermore, $x_{F_k^{O \cup A}} = x_{F_k^{w_s \cup A}}$, because $O$ and $w_s$ behave the same for words of length $\leq p_k(|x_{F_k^{O \cup A}}|)$. Consequently, V\ref{V4}($w_s,r_s$) is not satisfied, a contradiction to Claim \ref{claim:ws-is-rs-valid}. 

    \emph{To V\ref{V5}($O,r$)}: Suppose that there is an $i$ such that $\card{L_i^{O \cup A}} < \card{\{k \in \N^+ \mid (k,i) \in \dom(r)\}}$. By Claim \ref{claim:O-r-well-defined}, $\dom(r) = \N^+ \times \N^+$, hence $\card{L_i^{O \cup A}} = c$ for some $c \in \N$. Let $s$ be the step after which all tasks $\tau_{1,i}, \tau_{2,i}, \dots , \tau_{c+1,i}$ were treated. By V\ref{V5}($w_s,r_s$), it holds that
    \[\card{L_i^{w_s \cup A}} \geq \card{\{k \in \N^+ \mid (k,i) \in \dom(r_s)\}} \geq c+1.\]
    Since $O \sqsupseteq w_s$, $\card{L_i^{O \cup A}} \geq \card{L_i^{w_s \cup A}} \geq c+1$, a contradiction.
\end{claimproof}
\subsection{Results of the Oracle Construction}
For this section, let $S$ be the following statement:
\begin{center}
    $S \coloneqq \text{``}$All sets in $\PSPACE$ have p-optimal proof systems and there are arbitrarily complex sets having p-optimal proof systems.''
\end{center}
\begin{theorem}\label{thm:oracle-according-to-alg-1}
    Let $A \subseteq 2\N$ be arbitrary. There is an oracle $O$ constructed according to \emph{constructTALLY($\neg S,A$)}.
\end{theorem}
\begin{proof}
    Consider the oracle $O$ from Definition \ref{def:oracle-definition}. By Claims \ref{claim:O-r-well-defined} and \ref{claim:O-r-valid} and Lemma \ref{lem:properties-of-valid-oracle}, $O \in \TALLY$ and the statement $S$ holds relative to $O \cup A$, i.e., $\neg S$ fails relative to $O \cup A$. Note that $A \subseteq 2\N$ and $O \subseteq 2\N+1$, so in this case $\cup$ can be used as a marked union $\oplus$.
\end{proof}
\begin{corollary}
    There is an oracle $O \in \TALLY$ relative to which all sets in $\PSPACE^O$ have p-optimal proof systems and there are arbitrarily complex sets having p-optimal proof systems.
\end{corollary}
\begin{proof}
    Follows from Theorem \ref{thm:oracle-according-to-alg-1} for $A = \emptyset$.
\end{proof}
\begin{corollary}
   Let $S'$ be any statement of Theorems \ref{thm:bbs86}, \ref{thm:sparse-relativizing-results}, and \ref{thm:tally-relativizing-results}. Then there is an oracle relative to which both statements $S$ and $S'$ hold.
\end{corollary}
\begin{proof}
    All statements of Theorems \ref{thm:bbs86}, \ref{thm:sparse-relativizing-results}, and \ref{thm:tally-relativizing-results} are true either relative to $A_{\text{Yao}}$ where the polynomial-time hierarchy is infinite or $A_{\text{BGS}}$ where $\P = \PSPACE$. Without loss of generality, $A_{\text{Yao}}, A_{\text{BGS}} \subseteq 2\N$. Then by Theorem \ref{thm:sparse-relativization}, by a $\TALLY$-version of Theorem \ref{thm:sparse-relativization}, and by Theorem \ref{thm:oracle-according-to-alg-1}, there is an oracle relative to which $S'$ and $S$ hold.
\end{proof}
\begin{corollary}\label{cor:O2}
    There is an oracle $O_2$ relative to which $\PH^{O_2}$ is infinite, all sets in $\PSPACE^{O_2}$ have p-optimal proof systems and there are arbitrarily complex sets having p-optimal proof systems.
\end{corollary}
\begin{corollary}
    An infinite $\PH$ does not relativizably imply a negative answer to neither Q1, Q2, nor Q3.
\end{corollary}

\section{Conclusion and Open Questions}
The $\SPARSE$-relativization framework is a simple and useful tool to obtain improved proof barriers from sparse oracles. Using this framework, we are able to construct oracles that separate conjectures about complexity classes from conjectures about the existence of optimal and p-optimal proof systems.

A key component of the $\SPARSE$-relativization framework are the sparse-oracle-independent statements, as those are the statements we get ``for free'' when constructing suitable sparse oracles. Hence, we are interested in contributions to the following task.
\begin{align*}
&\text{\textbf{T1}: Find further statements that are tally-/sparse-oracle-independent}
\end{align*}
As explained in Section \ref{sec:sparse-relativization}, we obtain even stronger proof barriers by restricting the $\SPARSE$-relativization framework to $\SPARSE \cap \EXPH$. Hence, we also encourage to investigate the following task.
\begin{spreadlines}{0ex}
\begin{align*}
    &\text{\textbf{T2}:} \text{ Construct $\EXPH$-computable oracles relative to which sparse-oracle-independent} 
    \\
    \vspace{-100em}&\phantom{\text{\textbf{T2}:}} \text{ and tally-oracle-independent statements are true}
\end{align*}
\end{spreadlines}
An $\EXPH$-computable oracle relative to which the polynomial-time hierarchy is infinite is of special interest.

\bibliography{Literatursammlung-ArXiv}

\end{document}

%% file: 9-Standardnotation.tex
\def\P{\ensuremath{\mathrm{P}}}

\def\NP{\ensuremath{\mathrm{NP}}}

\def\E{\ensuremath{\mathrm{E}}}

\def\EE{\ensuremath{\mathrm{EE}}}
\def\NE{\ensuremath{\mathrm{NE}}}

\def\NEE{\ensuremath{\mathrm{NEE}}}
\def\FP{\ensuremath{\mathrm{FP}}}
\def\UP{\ensuremath{\mathrm{UP}}}
\def\DisjNP{\ensuremath{\mathrm{DisjNP}}}

\def\coNP{\ensuremath{\mathrm{coNP}}}

\def\coNE{\ensuremath{\mathrm{coNE}}}
\def\coNEE{\ensuremath{\mathrm{coNEE}}}

\def\TFNP{\ensuremath{\mathrm{TFNP}}}
\def\TALLY{\ensuremath{\mathrm{TALLY}}}
\def\SPARSE{\ensuremath{\mathrm{SPARSE}}}

\def\TAUT{\ensuremath{\mathrm{TAUT}}}
\def\SAT{\ensuremath{\mathrm{SAT}}}
\def\QBF{\ensuremath{\mathrm{QBF}}}

\def\TFNP{\ensuremath{\mathrm{TFNP}}}
\def\BPP{\ensuremath{\mathrm{BPP}}}
\def\RP{\ensuremath{\mathrm{RP}}}
\def\ZPP{\ensuremath{\mathrm{ZPP}}}

\def\PSPACE{\ensuremath{\mathrm{PSPACE}}}
\def\PH{\ensuremath{\mathrm{PH}}}

\def\DTIME{\ensuremath{\mathrm{DTIME}}}
\def\IP{\ensuremath{\mathrm{IP}}}
\def\EXPH{\ensuremath{\mathrm{EXPH}}}
\def\PP{\ensuremath{\mathrm{PP}}}
\def\CeqP{\ensuremath{\mathrm{C}\mathrm{=}\mathrm{P}}}
\def\modP{\ensuremath{\oplus\mathrm{P}}}

\def\N{\ensuremath{\mathrm{\mathbb{N}}}}

\def\psim{\ensuremath{\leq^\mathrm{p}}}

\DeclareMathOperator{\dom}{dom}
\DeclareMathOperator{\ran}{ran}

\def\sqsubsetneq{\mathrel{\sqsubseteq\kern-0.92em\raise-0.15em\hbox{\rotatebox{313}{\scalebox{1.1}[0.75]{\(\shortmid\)}}}\scalebox{0.3}[1]{\ }}}
\def\sqsupsetneq{\mathrel{\sqsupseteq\kern-0.92em\raise-0.15em\hbox{\rotatebox{313}{\scalebox{1.1}[0.75]{\(\shortmid\)}}}\scalebox{0.3}[1]{\ }}}
\newcommand{\card}[1]{\##1}

\newcommand{\xTrueThenElse}[2]{\if\x1{#1}\else{#2}\fi}



%% file: Literatursammlung-ArXiv.bib
@book{ab09,
author = {Arora, Sanjeev and Barak, Boaz},
title = {Computational Complexity: A Modern Approach},
year = {2009},
isbn = {0521424267},
publisher = {Cambridge University Press},
address = {USA},
edition = {1st}
}

@article{bbs86,
author = {Balc\'{a}zar, Jose L. and Book, Ronald V. and Sch\"{o}ning, Uwe},
title = {The polynomial-time hierarchy and sparse oracles},
year = {1986},
issue_date = {July 1986},
publisher = {Association for Computing Machinery},
address = {New York, NY, USA},
volume = {33},
number = {3},
issn = {0004-5411},
url = {https://doi.org/10.1145/5925.5937},
doi = {10.1145/5925.5937},
journal = {J. ACM},
month = may,
pages = {603–617},
numpages = {15}
}

@article{bcs95,
 author = {Daniel P. Bovet and Pierluigi Crescenzi and Riccardo Silvestri},
 doi = {https://doi.org/10.1006/jcss.1995.1030},
 issn = {0022-0000},
 journal = {Journal of Computer and System Sciences},
 number = {3},
 pages = {382-390},
 title = {Complexity Classes and Sparse Oracles},
 url = {https://www.sciencedirect.com/science/article/pii/S0022000085710306},
 volume = {50},
 year = {1995}
}

@article{bgs75,
  author    = {Theodore P. Baker and
               John Gill and
               Robert Solovay},
  title     = {Relativizations of the {P} =? {NP} Question},
  journal   = {{SIAM} Journal on Computing},
  volume    = {4},
  number    = {4},
  pages     = {431--442},
  year      = {1975},
  url       = {https://doi.org/10.1137/0204037},
  doi       = {10.1137/0204037},
  bibsource = {dblp computer science bibliography, https://dblp.org}
}

@article{bkm09,
 author = {Olaf Beyersdorff and Johannes Köbler and Jochen Messner},
 doi = {https://doi.org/10.1016/j.tcs.2009.05.021},
 issn = {0304-3975},
 journal = {Theoretical Computer Science},
 number = {38},
 pages = {3839-3855},
 title = {Nondeterministic functions and the existence of optimal proof systems},
 url = {https://www.sciencedirect.com/science/article/pii/S0304397509003843},
 volume = {410},
 year = {2009}
}

@article{bkm11,
 author = {Olaf Beyersdorff and Johannes Köbler and Sebastian Müller},
 doi = {https://doi.org/10.1016/j.ic.2010.11.006},
 issn = {0890-5401},
 journal = {Information and Computation},
 number = {3},
 pages = {320-332},
 title = {Proof systems that take advice},
 url = {https://www.sciencedirect.com/science/article/pii/S0890540110001902},
 volume = {209},
 year = {2011}
}

@techreport{bus96,
  author = {Sam R. Buss},
  title = {Lectures on Proof Theory},
  number = {SOCS 96.1},
  institution = {McGill University},
  year = {1996}
}

@article{boo94,
title = {On collapsing the polynomial-time hierarchy},
journal = {Information Processing Letters},
volume = {52},
number = {5},
pages = {235-237},
year = {1994},
issn = {0020-0190},
doi = {https://doi.org/10.1016/0020-0190(94)00157-X},
url = {https://www.sciencedirect.com/science/article/pii/002001909400157X},
author = {Ronald V. Book}
}

@article{cfm14,
 address = {New York, NY, USA},
 articleno = {7},
 author = {Chen, Yijia and Flum, J\"{o}rg and M\"{u}ller, Moritz},
 doi = {10.1145/2601336},
 issn = {1942-3454},
 issue_date = {May 2014},
 journal = {ACM Trans. Comput. Theory},
 month = {May},
 number = {2},
 numpages = {25},
 publisher = {Association for Computing Machinery},
 title = {Hard Instances of Algorithms and Proof Systems},
 url = {https://doi.org/10.1145/2601336},
 volume = {6},
 year = {2014}
}

@article{ck07,
 author = {Stephen A. Cook and Jan Krajíček},
 doi = {10.2178/jsl/1203350791},
 journal = {The Journal of Symbolic Logic},
 number = {4},
 pages = {1353--1371},
 publisher = {Association for Symbolic Logic},
 title = {Consequences of the Provability of {NP} $\subseteq$ {P/poly}},
 volume = {72},
 year = {2007}
}

@article{cr79,
 author = {Cook, Stephen A. and Reckhow, Robert A.},
 doi = {10.2307/2273702},
 journal = {Journal of Symbolic Logic},
 pages = {36-50},
 title = {The relative efficiency of propositional proof systems},
 volume = {44},
 year = {1979}
}

@InProceedings{deg24,
  author =	{Dingel, David and Egidy, Fabian and Gla{\ss}er, Christian},
  title =	{An Oracle with no {UP}-Complete Sets, but {NP = PSPACE}},
  booktitle =	{49th International Symposium on Mathematical Foundations of Computer Science (MFCS 2024)},
  pages =	{50:1--50:17},
  ISBN =	{978-3-95977-335-5},
  ISSN =	{1868-8969},
  year =	{2024},
  volume =	{306},
  publisher =	{Schloss Dagstuhl -- Leibniz-Zentrum f{\"u}r Informatik},
  address =	{Dagstuhl, Germany},
  URL =		{https://drops.dagstuhl.de/entities/document/10.4230/LIPIcs.MFCS.2024.50},
  URN =		{urn:nbn:de:0030-drops-206063},
  doi =		{10.4230/LIPIcs.MFCS.2024.50},
}

@inproceedings{dg20,
  author    = {Titus Dose and
               Christian Gla{\ss}er},
  editor    = {C. Paul and
               M. Bl{\"{a}}ser},
  title     = {{NP}-Completeness, Proof Systems, and Disjoint {NP}-Pairs},
  booktitle = {37th International Symposium on Theoretical Aspects of Computer Science,
               {STACS} 2020, March 10-13, 2020, Montpellier, France},
  series    = {LIPIcs},
  volume    = {154},
  pages     = {9:1--9:18},
  publisher = {Schloss Dagstuhl -- Leibniz-Zentrum f{\"{u}}r Informatik},
  year      = {2020},
  doi = {10.4230/LIPIcs.STACS.2020.9}
}

@article{dos20a,
  author    = {Titus Dose},
  title     = {An oracle separating conjectures about incompleteness in the finite
               domain},
  journal   = {Theoretical Computer Science},
  volume    = {809},
  pages     = {466--481},
  year      = {2020},
	doi = {10.1016/j.tcs.2020.01.003},
}

@article{dos20b,
 author = {Titus Dose},
 doi = {https://doi.org/10.1016/j.tcs.2020.09.040},
 issn = {0304-3975},
 journal = {Theoretical Computer Science},
 pages = {76-94},
 title = {Further oracles separating conjectures about incompleteness in the finite domain},
 url = {https://www.sciencedirect.com/science/article/pii/S030439752030551X},
 volume = {847},
 year = {2020}
}

@InProceedings{eeg22,
  author =	{Ehrmanntraut, Anton and Egidy, Fabian and Gla{\ss}er, Christian},
  title =	{Oracle with {P = NP} $\cap$ {coNP}, but No Many-One Completeness in {UP}, {DisjNP}, and {DisjCoNP}},
  booktitle =	{47th International Symposium on Mathematical Foundations of Computer Science (MFCS 2022)},
  pages =	{45:1--45:15},
  series =	{Leibniz International Proceedings in Informatics (LIPIcs)},
  ISBN =	{978-3-95977-256-3},
  ISSN =	{1868-8969},
  year =	{2022},
  volume =	{241},
  editor =	{Szeider, S. and Ganian, R. and Silva, A.},
  publisher =	{Schloss Dagstuhl -- Leibniz-Zentrum f{\"u}r Informatik},
  address =	{Dagstuhl, Germany},
  URL =		{https://drops.dagstuhl.de/opus/volltexte/2022/16843},
  URN =		{urn:nbn:de:0030-drops-168435},
  doi =		{10.4230/LIPIcs.MFCS.2022.45}
}

@inproceedings{eg25,
author = {Egidy, Fabian and Gla\ss{}er, Christian},
title = {Optimal Proof Systems for Complex Sets Are Hard to Find},
year = {2025},
isbn = {9798400715105},
publisher = {Association for Computing Machinery},
address = {New York, NY, USA},
url = {https://doi.org/10.1145/3717823.3718182},
doi = {10.1145/3717823.3718182},
booktitle = {Proceedings of the 57th Annual ACM Symposium on Theory of Computing},
pages = {1329–1340},
numpages = {12},
location = {Prague, Czechia},
series = {STOC '25}
}

@Article{for99,
  author =   {Lance Fortnow},
  title =    {Relativized worlds with an infinite hierarchy},
  journal =  {Information Processing Letters},
  year =     1999,
  volume =   69,
  number =   4,
  pages =    {309-313}
}

@article{fs88,
title = {Are there interactive protocols for {co-NP} languages?},
journal = {Information Processing Letters},
volume = {28},
number = {5},
pages = {249-251},
year = {1988},
issn = {0020-0190},
doi = {https://doi.org/10.1016/0020-0190(88)90199-8},
url = {https://www.sciencedirect.com/science/article/pii/0020019088901998},
author = {Lance Fortnow and Michael Sipser}
}

@Article{gssz04,
  author =       {Christian Gla{\ss}er and Alan L. Selman and Samik Sengupta and Liyu Zhang},
  title =        {Disjoint {NP}-Pairs},
  journal =      {SIAM Journal on Computing},
  year =         2004,
  volume =       33,
  number =       6,
  pages =        {1369-1416},
	doi = {10.1137/S0097539703425848},
}

@incollection{has89,
  author    = {H{\aa}stad, Johan},
  title     = {Almost Optimal Lower Bounds for Small Depth Circuits},
  booktitle = {Randomness and Computation},
  editor    = {Micali, S.},
  series    = {Advances in Computing Research},
  volume    = {5},
  pages     = {143--170},
  publisher = {JAI Press},
  address   = {Greenwich},
  year      = {1989}
}

@InProceedings{hlr23,
  author =	{Hirahara, Shuichi and Lu, Zhenjian and Ren, Hanlin},
  title =	{Bounded Relativization},
  booktitle =	{38th Computational Complexity Conference (CCC 2023)},
  pages =	{6:1--6:45},
  series =	{Leibniz International Proceedings in Informatics (LIPIcs)},
  ISBN =	{978-3-95977-282-2},
  ISSN =	{1868-8969},
  year =	{2023},
  volume =	{264},
  editor =	{Ta-Shma, Amnon},
  publisher =	{Schloss Dagstuhl -- Leibniz-Zentrum f{\"u}r Informatik},
  address =	{Dagstuhl, Germany},
  URL =		{https://drops.dagstuhl.de/entities/document/10.4230/LIPIcs.CCC.2023.6},
  URN =		{urn:nbn:de:0030-drops-182764},
  doi =		{10.4230/LIPIcs.CCC.2023.6}
}

@article{hrst17,
author = {H\r{a}stad, Johan and Rossman, Benjamin and Servedio, Rocco A. and Tan, Li-Yang},
title = {An Average-Case Depth Hierarchy Theorem for Boolean Circuits},
year = {2017},
issue_date = {October 2017},
publisher = {Association for Computing Machinery},
address = {New York, NY, USA},
volume = {64},
number = {5},
issn = {0004-5411},
url = {https://doi.org/10.1145/3095799},
doi = {10.1145/3095799},
journal = {Journal of the ACM},
month = aug,
articleno = {35},
numpages = {27}
}

@article{hss21,
author = {Hitchcock, John M. and Sekoni, Adewale and Shafei, Hadi},
title = {Polynomial-Time Random Oracles and Separating Complexity Classes},
year = {2021},
issue_date = {March 2021},
publisher = {Association for Computing Machinery},
address = {New York, NY, USA},
volume = {13},
number = {1},
issn = {1942-3454},
url = {https://doi.org/10.1145/3434389},
doi = {10.1145/3434389},
journal = {ACM Trans. Comput. Theory},
month = jan,
articleno = {1},
numpages = {6}
}

@article{kha22,
 author = {Khaniki, Erfan},
 doi = {10.1017/jsl.2021.99},
 journal = {The Journal of Symbolic Logic},
 number = {3},
 pages = {912–937},
 publisher = {Cambridge University Press},
 title = {New relations and separations of conjectures about incompleteness in the finite domain},
 volume = {87},
 year = {2022}
}

@inproceedings{kl80,
 address = {New York, NY, USA},
 author = {Karp, Richard M. and Lipton, Richard J.},
 booktitle = {Symposium on Theory of Computing},
 doi = {10.1145/800141.804678},
 isbn = {0897910176},
 location = {Los Angeles, California, USA},
 numpages = {8},
 pages = {302–309},
 publisher = {ACM},
 series = {STOC '80},
 title = {Some connections between nonuniform and uniform complexity classes},
 url = {https://doi.org/10.1145/800141.804678},
 year = {1980}
}

@inproceedings{km98,
 author = {Johannes K{\"{o}}bler and
Jochen Messner},
 bibsource = {dblp computer science bibliography, https://dblp.org},
 booktitle = {Proceedings of the 13th Annual {IEEE} Conference on Computational
Complexity},
 doi = {10.1109/CCC.1998.694599},
 pages = {132--140},
 title = {Complete problems for promise classes by optimal proof systems for
test sets},
 url = {https://doi.org/10.1109/CCC.1998.694599},
 year = {1998}
}

@article{kmt03,
 author = {Johannes Köbler and Jochen Messner and Jacobo Torán},
 doi = {https://doi.org/10.1016/S0890-5401(03)00058-0},
 issn = {0890-5401},
 journal = {Information and Computation},
 number = {1},
 pages = {71-92},
 title = {Optimal proof systems imply complete sets for promise classes},
 url = {https://www.sciencedirect.com/science/article/pii/S0890540103000580},
 volume = {184},
 year = {2003}
}

@article{ko89,
author = {Ko, Ker-I},
title = {Relativized Polynomial Time Hierarchies Having Exactly K Levels},
journal = {SIAM Journal on Computing},
volume = {18},
number = {2},
pages = {392-408},
year = {1989},
doi = {10.1137/0218027},
URL = {https://doi.org/10.1137/0218027}
}

@article{kp89,
 author = {Jan Krajíček and Pavel Pudl{\'a}k},
 doi = {10.2307/2274765},
 journal = {Journal of Symbolic Logic},
 pages = {1063-1079},
 title = {Propositional proof systems, the consistency of
first order theories and the complexity of computations},
 volume = {54},
 year = {1989}
}

@article{kstt92,
title = {Turing machines with few accepting computations and low sets for {PP}},
journal = {Journal of Computer and System Sciences},
volume = {44},
number = {2},
pages = {272-286},
year = {1992},
issn = {0022-0000},
doi = {https://doi.org/10.1016/0022-0000(92)90022-B},
url = {https://www.sciencedirect.com/science/article/pii/002200009290022B},
author = {Johannes Köbler and Uwe Schöning and Seinosuke Toda and Jacobo Torán}
}

@article{ls86,
 address = {New York, NY, USA},
 author = {Long, Timothy J. and Selman, Alan L.},
 doi = {10.1145/5925.5938},
 issn = {0004-5411},
 issue_date = {July 1986},
 journal = {Journal of the ACM},
 number = {3},
 numpages = {10},
 pages = {618–627},
 publisher = {Association for Computing Machinery},
 title = {Relativizing complexity classes with sparse oracles},
 url = {https://doi.org/10.1145/5925.5938},
 volume = {33},
 year = {1986}
}

@inbook{lut97,
 address = {Berlin, Heidelberg},
 author = {Lutz, Jack H.},
 booktitle = {Complexity Theory Retrospective II},
 isbn = {0387949739},
 numpages = {36},
 pages = {225–260},
 publisher = {Springer-Verlag},
 title = {The quantitative structure of exponential time},
 year = {1997}
}

@phdthesis{mes00,
 author = {Jochen Messner},
 school = {Universit\"at Ulm},
 title = {On the Simulation Order of Proof Systems},
 year = {2000}
}

@inproceedings{mes99,
 address = {Berlin, Heidelberg},
 author = {Jochen Messner},
 booktitle = {Symposium on Theoretical Aspects of Computer Science},
 doi = {10.1007/3-540-49116-3_51},
 pages = {541--550},
 publisher = {Springer-Verlag},
 title = {On Optimal Algorithms and Optimal Proof Systems},
 url = {https://api.semanticscholar.org/CorpusID:7759364},
 year = {1999}
}

@article{pud17,
 author = {Pavel Pudl{\'a}k},
 doi = {doi:10.1017/bsl.2017.32},
 journal = {The Bulletin of Symbolic Logic},
 number = {4},
 pages = {405--441},
 publisher = {[Association for Symbolic Logic, Cambridge University Press]},
 title = {Incompleteness in the Finite Domain},
 volume = {23},
 year = {2017}
}

@techreport{raz94,
  author      = {Alexander Razborov},
  title       = {On Provably Disjoint {NP}-Pairs},
  institution = {Electronic Colloquium on Computational Complexity},
  year        = {1994},
  number      = {TR94-006},
  url         = {https://eccc.weizmann.ac.il/report/1994/006/}
}

@inproceedings{sad97,
 author = {Zenon Sadowski},
 bibsource = {dblp computer science bibliography, https://dblp.org},
 booktitle = {Fundamentals of Computation Theory},
 doi = {10.1007/BFB0036203},
 pages = {423--428},
 publisher = {Springer},
 title = {On an optimal quantified propositional proof system and a complete
language for {NP} $\cap$ {coNP}},
 url = {https://doi.org/10.1007/BFb0036203},
 volume = {1279},
 year = {1997}
}

@inproceedings{sad99,
 address = {Berlin, Heidelberg},
 author = {Sadowski, Zenon},
 booktitle = {Computer Science Logic},
 doi = {https://doi.org/10.1007/10703163_13},
 isbn = {978-3-540-48855-2},
 pages = {179--187},
 publisher = {Springer-Verlag},
 title = {On an Optimal Deterministic Algorithm for SAT},
 year = {1999}
}

@article{sha92,
author = {Shamir, Adi},
title = {{IP = PSPACE}},
year = {1992},
issue_date = {Oct. 1992},
publisher = {Association for Computing Machinery},
address = {New York, NY, USA},
volume = {39},
number = {4},
issn = {0004-5411},
url = {https://doi.org/10.1145/146585.146609},
doi = {10.1145/146585.146609},
journal = {Journal of the ACM},
month = oct,
pages = {869–877},
numpages = {9}
}

@phdthesis{sim77,
 author = {Simon, István},
 school = {Stanford University},
 title = {On Some Subrecursive Reducibilities.},
 year = {1977}
}

@inproceedings{yao85,
 address = {USA},
 author = {Yao, Andrew},
 booktitle = {Proceedings of the 26th Annual Symposium on Foundations of Computer Science},
 doi = {10.1109/SFCS.1985.49},
 isbn = {0818608444},
 numpages = {10},
 pages = {1–10},
 publisher = {IEEE Computer Society},
 series = {SFCS '85},
 title = {Separating the polynomial-time hierarchy by oracles},
 url = {https://doi.org/10.1109/SFCS.1985.49},
 year = {1985}
}
